\makeatletter
\def\input@path{{styles/}{../styles/}}
\makeatother

\ifx\GenSoCGVer\undefined%
\documentclass[12pt]{article}%
\newcommand{\SoCG}[1]{}
\newcommand{\NotSoCG}[1]{#1}%
\else
\newcommand{\SoCG}[1]{#1}
\newcommand{\NotSoCG}[1]{}%
\documentclass[english]{socg-lipics-v2021}
\fi

\NotSoCG{%
   \def\UseBibLatex{1}%
}

\ifx\sarielComp\undefined%
\newcommand{\SarielComp}[1]{}
\newcommand{\NotSarielComp}[1]{#1}%
\else
\newcommand{\SarielComp}[1]{#1}%
\newcommand{\NotSarielComp}[1]{}%
\fi

\NotSoCG{%
   \setlength{\marginparwidth}{6cm}%
}
\usepackage[table]{xcolor}%
\SoCG{\usepackage{todonotes}}
\NotSoCG{%
   \usepackage{todonotes}
   \usepackage[in]{fullpage}%
}

\NotSoCG{%
   \usepackage{lmodern}
}
\usepackage{amsmath}%
\usepackage{amssymb}%

\SarielComp{\usepackage{sariel_colors}}%

\NotSoCG{%
   \usepackage[amsmath,thmmarks]{ntheorem}%
   \theoremseparator{.}%

   \usepackage{titlesec}%
   \titlelabel{\thetitle. }%
}

\usepackage{xcolor}%
\usepackage{mleftright}%
\usepackage{xspace}%
\usepackage{graphicx}
\usepackage{hyperref}%
\NotSoCG{%
   \usepackage{euscript}%
}

\usepackage{subcaption}

\hypersetup{%
      unicode,
      breaklinks,%
      colorlinks=true,%
      urlcolor=[rgb]{0.25,0.0,0.0},%
      linkcolor=[rgb]{0.5,0.0,0.0},%
      citecolor=[rgb]{0,0.2,0.445},%
      filecolor=[rgb]{0,0,0.4},
      anchorcolor=[rgb]={0.0,0.1,0.2}%
}
\usepackage[ocgcolorlinks]{ocgx2}

\numberwithin{figure}{section}%
\numberwithin{table}{section}%
\numberwithin{equation}{section}%

\NotSoCG{%
\theoremseparator{.}%

\theoremstyle{plain}%
\newtheorem{theorem}{Theorem}[section]

\newtheorem{lemma}[theorem]{Lemma}

\newtheorem{corollary}[theorem]{Corollary}

\newtheorem{observation}[theorem]{Observation}

\theoremstyle{plain}%
\theoremheaderfont{\sf} \theorembodyfont{\upshape}%
\newtheorem*{remark:unnumbered}[theorem]{Remark}%
\newtheorem{remark}[theorem]{Remark}%
\newtheorem{algorithm}[theorem]{Algorithm}%
\newtheorem{definition}[theorem]{Definition}

\newcommand{\myqedsymbol}{\rule{2mm}{2mm}}

\theoremheaderfont{\em}%
\theorembodyfont{\upshape}%
\theoremstyle{nonumberplain}%
\theoremseparator{}%
\theoremsymbol{\myqedsymbol}%
\newtheorem{proof}{Proof:}%

}

\SoCG{%
   \newtheorem{algorithm}[theorem]{Algorithm}
}

\definecolor{blue25emph}{rgb}{0, 0, 11}
\providecommand{\emphic}[2]{%
   \textcolor{blue25emph}{%
      \textbf{\emph{#1}}}%
   \index{#2}}

\providecommand{\emphi}[1]{\emphic{#1}{#1}}

\definecolor{almostblack}{rgb}{0, 0, 0.3}
\providecommand{\emphw}[1]{{\textcolor{almostblack}{\emph{#1}}}}%

\newcommand{\atgen}{\symbol{'100}}%
\newcommand{\SarielThanks}[1]{%
   \thanks{%
      Department of Computer Science; %
      University of Illinois; %
      201 N. Goodwin Avenue; %
      Urbana, IL, 61801, USA; %
      \href{mailto:spam@illinois.edu}{sariel@illinois.edu}; %
      \url{http://sarielhp.org/}. %
   #1%
   }%
}
\newcommand{\BenThanks}[1]{%
   \thanks{Department of Computer Science;
           University of Texas at Dallas; Richardson, TX 75080, USA;
      {\tt benjamin.raichel\atgen{}utdallas.edu}; {\tt
         \url{http://utdallas.edu/\string~benjamin.raichel}.} #1}}

\newcommand{\EliotThanks}[1]{%
   \thanks{%
      Department of Computer Science; %
      University of Illinois; %
      201 N. Goodwin Avenue; %
      Urbana, IL, 61801, USA; %
      \href{mailto:erobson2@illinois.edu}{erobson2@illinois.edu}; %
      \url{https://eliotwrobson.github.io}.%
   #1%
   }%
}

\newcommand{\HLink}[2]{\hyperref[#2]{#1~\ref*{#2}}}
\newcommand{\HLinkSuffix}[3]{\hyperref[#2]{#1\ref*{#2}{#3}}}

\newcommand{\figlab}[1]{\label{fig:#1}}
\newcommand{\figref}[1]{\HLink{Figure}{fig:#1}}

\newcommand{\thmlab}[1]{{\label{theo:#1}}}

\newcommand{\corlab}[1]{\label{cor:#1}}
\newcommand{\corref}[1]{\HLink{Corollary}{cor:#1}}%

\providecommand{\deflab}[1]{\label{def:#1}}
\newcommand{\defref}[1]{\HLink{Definition}{def:#1}}
\newcommand{\defrefY}[2]{\hyperref[def:#2]{#1}}

\newcommand{\alglab}[1]{\label{Algorithm:#1}}%
\newcommand{\algref}[1]{\HLink{Algorithm}{Algorithm:#1}}%

\newcommand{\obslab}[1]{\label{observation:#1}}
\newcommand{\obsref}[1]{\HLink{Observation}{observation:#1}}

\newcommand{\remlab}[1]{\label{rem:#1}}
\newcommand{\remref}[1]{\HLink{Remark}{rem:#1}}%

\newcommand{\lemlab}[1]{\label{lemma:#1}}
\newcommand{\lemref}[1]{\HLink{Lemma}{lemma:#1}}%

\newcommand{\seclab}[1]{\label{sec:#1}}
\newcommand{\secref}[1]{\HLink{Section}{sec:#1}}

\providecommand{\eqlab}[1]{}%
\renewcommand{\eqlab}[1]{\label{equation:#1}}
\newcommand{\Eqref}[1]{\HLinkSuffix{Eq.~(}{equation:#1}{)}}

\usepackage{stmaryrd}%
\providecommand{\IntRange}[1]{\mleft\llbracket #1 \mright\rrbracket}
\newcommand{\IRX}[1]{\IntRange{#1}}%
\newcommand{\IRY}[2]{\left\llbracket #1:#2 \right\rrbracket}

\newcommand{\remove}[1]{}%
\newcommand{\Set}[2]{\left\{ #1 \;\middle\vert\; #2 \right\}}
\newcommand{\pth}[1]{\mleft( #1 \mright)}%

\newcommand{\floor}[1]{\left\lfloor {#1} \right\rfloor}

\newcommand{\cardin}[1]{\left| {#1} \right|}%

\renewcommand{\th}{th\xspace}

\renewcommand{\Re}{\mathbb{R}}%

\usepackage[inline]{enumitem}

\newlist{compactenumA}{enumerate}{5}%
\setlist[compactenumA]{topsep=0pt,itemsep=-1ex,partopsep=1ex,parsep=1ex,%
   label=(\Alph*)}%

\newlist{compactenuma}{enumerate}{5}%
\setlist[compactenuma]{topsep=0pt,itemsep=-1ex,partopsep=1ex,parsep=1ex,%
   label=(\alph*)}%

\newlist{compactenumI}{enumerate}{5}%
\setlist[compactenumI]{topsep=0pt,itemsep=-1ex,partopsep=1ex,parsep=1ex,%
   label=(\Roman*)}%

\newlist{compactenumi}{enumerate}{5}%
\setlist[compactenumi]{topsep=0pt,itemsep=-1ex,partopsep=1ex,parsep=1ex,%
   label=(\roman*)}%

\newlist{compactitem}{itemize}{5}%
\setlist[compactitem]{topsep=0pt,itemsep=-1ex,partopsep=1ex,parsep=1ex,%
   label=\ensuremath{\bullet}}%

\providecommand{\BibLatexMode}[1]{}
\providecommand{\BibTexMode}[1]{#1}

\ifx\UseBibLatex\undefined%
  \renewcommand{\BibLatexMode}[1]{}
  \renewcommand{\BibTexMode}[1]{#1}
\else
  \renewcommand{\BibLatexMode}[1]{#1}
  \renewcommand{\BibTexMode}[1]{}
\fi

\newcommand{\UsePackage}[1]{%
  \IfFileExists{styles/#1.sty}{%
      \usepackage{styles/#1}%
   }{%
      \IfFileExists{../styles/#1.sty}{%
         \usepackage{../styles/#1}%
      }{%
         \usepackage{#1}%
      }%
   }%
}

\UsePackage{mathcalb}

\BibLatexMode{%
   \usepackage[bibencoding=utf8,style=alphabetic,backend=biber]{biblatex}%
   \UsePackage{sariel_biblatex}%
}

\providecommand{\Mh}[1]{#1}%

\newcommand{\etal}{\textit{et~al.}\xspace}
\newcommand{\Frechet}{Fr\'{e}chet\xspace}
\newcommand{\norm}[1]{\left\| {#1} \right\|}

\newcommand{\dY}[2]{\norm{#1 #2}}
\newcommand{\dSY}[2]{\mathcalb{d}\pth{#1, #2}}
\newcommand{\nnY}[2]{\Downarrow_{#2}\!\!\pth{#1}}

\newcommand{\curveA}{\Mh{\pi}}%

\newcommand{\curveB}{\Mh{\sigma}}%

\newcommand{\cA}{\curveA}
\newcommand{\cB}{\curveB}
\newcommand{\cAs}{\overline{\cA}}
\newcommand{\cBs}{\overline{\cB}}

\newcommand{\simpX}[1]{\mathrm{s{i}m{p}l}\pth{#1}}

\newcommand{\Cell}{\Mh{C}}
\newcommand{\CellY}[2]{\Cell_{#1, #2}}

\newcommand{\edge}{\Mh{e}}

\newcommand{\vecP}{\mathcalb{p}}%
\newcommand{\vecQ}{\mathcalb{q}}

\newcommand{\curveC}{\Mh{\tau}}%

\newcommand{\lenX}[1]{\left\|{#1} \right\|}

\newcommand{\pA}{\Mh{p}}%
\newcommand{\pB}{\Mh{q}}%
\newcommand{\pC}{\Mh{t}}%
\newcommand{\pD}{\Mh{z}}%
\newcommand{\HeightChar}{\Mh{\mathcalb{h}}}

\newcommand{\heightX}[1]{\HeightChar^{}\pth{#1}}
\newcommand{\FrChar}{{\EuScript{F}}}
\newcommand{\distC}{\mathsf{d}}%
\newcommand{\distFr}[2]{\distC^{}_\FrChar\pth{#1, #2}}
\newcommand{\distDFr}[2]{{\overline{\mathsf{d}^{}_\FrChar}\pth{#1, #2}}}
\newcommand{\distWFr}[2]{\mathsf{d}^{w}_{\EuScript{F}}\pth{#1, #2}}
\newcommand{\distVEFr}[2]{\mathsf{d}^{ve}_{\EuScript{F}}\pth{#1, #2}}

\newcommand{\dSWY}[2]{\distC_{SW}\pth{#1,#2}}

\newcommand{\DotProd}[2]{\permut{{#1},{#2}}}
\newcommand{\permut}[1]{\left\langle {#1} \right\rangle}

\newcommand{\pp}{\mathcalb{p}}%
\newcommand{\qq}{\mathcalb{q}}%

\newcommand{\pps}{\mathcalb{p}_0}%
\newcommand{\qqs}{\mathcalb{q}_0}%

\newcommand{\ppe}{\mathcalb{p}_1}%
\newcommand{\qqe}{\mathcalb{q}_1}%

\newcommand{\uus}{\mathcalb{u}_0}%

\newcommand{\wC}{\Mh{\mathcalb{w}}}
\newcommand{\wX}[1]{\wC\pth{#1}}

\newcommand{\WidthChar}{\omega}
\newcommand{\WidthX}[1]{\WidthChar\pth{#1}}
\newcommand{\widthY}[2]{{\WidthChar}^{}_{#2}\pth{#1}}

\newcommand{\Grid}{H}

\newcommand{\Term}[1]{\textsf{#1}}

\newcommand{\DAG}{\Term{DAG}\xspace}%

\newcommand{\G}{\Mh{G}}%
\newcommand{\VV}{\Mh{V}}%
\newcommand{\EE}{\Mh{E}}%

\newcommand{\btl}{\Mh{\mathcalb{b}}}%

\providecommand{\VV}{\Mh{\mathsf{V}}}%
\newcommand{\VX}[1]{\VV\pth{#1}}
\newcommand{\EGX}[1]{\EE\pth{#1}}

\newcommand{\baseUrlX}[1]{https://frechet.xyz/examples/#1}
\newcommand{\doggy}{\Mh{R}}%
\newcommand{\doggyY}[2]{\doggy\pth{#1, #2} }

\newcommand{\elevC}{\Mh{\mathcalb{e}}}%
\newcommand{\elevCS}{\Mh{\mathcalb{e}'}}%
\newcommand{\elevD}{\Mh{\overline{\mathcalb{e}}}}%
\newcommand{\elevY}[2]{\elevC\pth{#1,#2}}
\newcommand{\elevX}[1]{\elevC\pth{#1}}
\newcommand{\elevSY}[2]{\elevCS\pth{#1,#2}}
\newcommand{\elevDY}[2]{\elevD\pth{#1,#2}}

\newcommand{\repX}[1]{\mathrm{rep}\pth{#1}}%

\newcommand{\seg}{\Mh{\tau}}
\newcommand{\segA}{\Mh{\nu}}

\newcommand{\monoX}[1]{\mathrm{mono}\pth{#1}}%

\newcommand{\wFX}[1]{\mathcalb{w}_{\FrChar}\pth{#1}}

\newcommand{\hippoX}[1]{\mathcalb{h}\pth{#1}}%

\newcommand{\lb}{\mathcalb{l}}%
\newcommand{\lbY}[2]{\lb\pth{#1, #2}}%
\newcommand{\slackX}[1]{\Delta\pth{#1}}%

\newcommand{\mrp}{\mathcalb{m}}%
\newcommand{\MrpY}[2]{\mathcal{M}_{#1,#2}}%
\newcommand{\MrpMY}[2]{\mathcal{M}^+_{#1,#2}}%

\newcommand{\urlBaseX}[1]{https://frechet.xyz/examples/#1}
\newcommand{\eps}{\varepsilon}

\newcommand{\VE}{\textsf{VE}\xspace}
\newcommand{\VEFrechet}{\VE{}-\Frechet{}\xspace}%
\newcommand{\GVEC}{\G_{\VE}}%
\newcommand{\GVEY}[2]{\GVEC \pth{#1, #2}}
\newcommand{\VEmorphing}{\mathcalb{m}_\VE}%
\newcommand{\morphVE}[2]{\mrp_\VE\pth{#1,#2}}%

\newcommand{\relCC}{\mathsf{N}}
\newcommand{\relCZ}[3]{\mathsf{N}_{#1}\pth{#2, #3}}

\newcommand{\CDTW}{\Term{CDTW}\xspace}
\newcommand{\costX}[1]{\mathrm{cost}\pth{#1}}
\newcommand{\dX}[1]{\mathop{\mathrm{d}}\!#1}

\newcommand{\dTWY}[2]{\distC_{\CDTW}\pth{#1,#2}}
\newcommand{\asinh}{\mathop{\mathrm{asinh}}}

\newcommand{\pnt}{\mathsf{p}}
\newcommand{\pntA}{\mathsf{q}}
\newcommand{\signX}[1]{\mathrm{sign}\pth{#1}}

\newcommand{\lbSDY}[2]{\ell\pth{#1,#2}}%
\newcommand{\Python}{\texttt{Python}\xspace}
\newcommand{\Julia}{\texttt{Julia}\xspace}
\newcommand{\CPP}{\texttt{C++}\xspace}

\providecommand{\si}[1]{#1}
\renewcommand{\si}[1]{#1}
\newcommand{\Cells}{\Xi}
\DefineNamedColor{named}{AlgorithmColor}{cmyk}{0.07,0.90,0,0.34}

\newcommand{\AlgorithmI}[1]{{%
      \textcolor[named]{AlgorithmColor}{\texttt{\bf{#1}}}%
   }}
\newcommand{\Alg}{\AlgorithmI{comp{}Profile}\xspace}
\newcommand{\extract}{\AlgorithmI{extract}\xspace}
\newcommand{\dirY}[2]{\overrightarrow{ #1 #2}}
\newcommand{\priceX}[1]{\mathrm{price}\pth{#1}}%

\SoCG{\newcommand{\myparagraph}[1]{\subparagraph*{#1}}}%
\NotSoCG{\newcommand{\myparagraph}[1]{\paragraph*{#1}}}%

\SoCG{%
\EventEditors{Oswin Aichholzer and Haitao Wang}
\EventNoEds{2}
\EventLongTitle{41st International Symposium on Computational Geometry (SoCG
2025)}
\EventShortTitle{SoCG 2025}
\EventAcronym{SoCG}
\EventYear{2025}
\EventDate{June 23--27, 2025}
\EventLocation{Kanazawa, Japan}
\EventLogo{socg-logo.pdf}
\SeriesVolume{332}
\ArticleNo{XX}

}

\BibLatexMode{%
   \bibliography{fmq} }

\title{The \Frechet Distance Unleashed: Approximating a Dog with a Frog}%

\NotSoCG{%
   \author{%
      Sariel Har-Peled%
      \SarielThanks{Work on this paper was partially supported by NSF AF award CCF-2317241.}%
      \and%
      Benjamin Raichel%
      \BenThanks{Work on this paper was partially supported by NSF CCF award 2311179.}%
      \and%
      Eliot W. Robson%
      \EliotThanks{}%
   }%
}%

\SoCG{%
   \author{Sariel Har-Peled}%
   {Department of Computer Science, University of Illinois, 201 N. Goodwin Avenue, Urbana, IL 61801, USA %
      \and \url{https://sarielhp.org}%
   }%
   {sariel@illinois.edu}%
   {https://orcid.org/0000-0003-2638-9635}%
   {Work on this paper was partially supported by an NSF AF award CCF-2317241.}%

   \author{Benjamin Raichel}{Department of Computer Science; University of Texas at Dallas; 800 W. Campbell Road; Richardson, TX 75080, USA \and \url{http_//utdallas.edu/\string~benjamin.raichel} } {benjamin.raichel@utdallas.edu}%
   {{https://orcid.org/0000-0001-6584-4843}}%
   {Work on this paper was partially supported by NSF AF Award CCF-2311179.}

   \author{Eliot W. Robson}%
   {Department of Computer Science; University of Illinois; 201 N. Goodwin Avenue; Urbana, IL, 61801, USA \and \url{https://eliotwrobson.github.io/}%
   }%
   {erobson2@illinois.edu} {https://orcid.org/0000-0002-1476-6715}%
   {Work on this paper was partially supported by an NSF AF award CCF-2317241.}%

   \authorrunning{S. Har-Peled, B. Raichel, and E. W. Robson}

   \Copyright{Sariel Har-Peled, Benjamin Raichel, and Eliot W. Robson}

   \ccsdesc[500]{Theory of computation~Computational geometry}%

   \keywords{Curve similarity, \Frechet distance}%

   \category{} %

}

\date{\today}

\begin{document}

\SoCG{%
   \relatedversion{\url{https://arxiv.org/abs/2407.03101}}%
}%
\maketitle

\begin{abstract}
    We show that a variant of the continuous \Frechet{} distance between polygonal curves can be computed using essentially the same algorithm used to solve the discrete version. The new variant is not necessarily monotone, but this shortcoming can be easily handled via refinement.

    Combined with a Dijkstra/Prim type algorithm, this leads to a realization of the \Frechet distance (i.e., a morphing) that is locally optimal (aka locally correct), that is both easy to compute, and in practice, takes near linear time on many inputs.  The new morphing has the property that the leash is always as short as possible.  These matchings/morphings are more natural and are better than the ones computed by standard algorithms -- in particular, they handle noise more graciously.  This approach should make the \Frechet distance more useful for real-world applications.

    We implemented the new algorithm and various strategies to obtain reasonably fast practical performance. We performed extensive experiments on our new algorithm, and released publicly available (and easily installable and usable) \Julia and \Python packages. Our algorithms can be used to compute the almost-exact \Frechet distance between polygonal curves.

    Implementations and numerous examples are available here: \href{https://frechet.xyz}{frechet.xyz}.

    We emphasize, however, that the existing state-of-the-art algorithm/implementation in \CPP is faster, by several orders of magnitude, than our current algorithm/implementation.
\end{abstract}

\section{Introduction}

\subsection{Definitions}

Given two polygonal curves, their \Frechet distance is the length of a leash that a person needs if they walk along one of the curves. In contrast, a dog connected by the leash walks along the other curve, assuming they synchronize their walks to minimize the length of this leash. (I.e.\ they walk to minimize their maximum distance apart during the walk.)  Our approach is slightly different than the standard approach, and we define it carefully first.

\subsubsection{Free space diagram and morphings}

\begin{definition}
    \deflab{u_parameter}%
    For a (directed) curve $\cA \subseteq \Re^d$, its \emphi{uniform parameterization} is the bijection $\cA:[0,\lenX{\cA}] \rightarrow \cA$, where $\lenX{\cA}$ is the length of $\cA$, and for any $x \in [0,\lenX{\cA}]$, the point $\cA(x)$ is at distance $x$ (along $\cA$) from the starting point of $\cA$.
\end{definition}

\begin{definition}
    \deflab{f_s_d}%
    The \emphw{free space diagram} of two curves $\curveA$ and $\curveB$ is the rectangle $\doggy =\doggyY{\curveA}{\curveB} = [0, \lenX{\curveA}] \times [0, \lenX{\curveB}]$. Specifically, for any point $(x,y) \in \doggy$, we associate the \emphi{elevation function} $\elevY{x}{y} = \dY{\curveA(x)}{-\curveB(y)}$.
\end{definition}

The free space diagram $\doggy$ is partitioned into a non-uniform grid, where each cell corresponds to all leash lengths when a point lies on a fixed segment of one curve, and the other lies on a fixed segment of the other curve, see \figref{Frechet_ex}.  For a given value $\delta \geq 0$, the \emph{sublevel set} of a real-valued function consists of all inputs whose function value is $\leq \delta$. It is known that for any value $\delta \geq 0$, the sublevel set of the elevation function inside such a grid cell is a clipped ellipse.

\begin{definition}
    A \emphi{morphing}\footnote{A morphing induces a natural homotopy between the two curves.} $\mrp$ between $\cA$ and $\cB$ is a (not self-intersecting) curve $\mrp \subseteq \doggyY{\curveA}{\curveB}$ with endpoints $(0,0)$ and $(\lenX{\cA},\lenX{\cB})$. The set of all morphings between $\cA$ and $\cB$ is $\MrpY{\cA}{\cB}$.  A morphing that is a segment inside each cell of the free space diagram that it visits, is \emphw{well behaved}\footnote{All the morphings we deal with in this paper are well behaved.  }.
\end{definition}

Intuitively, a morphing is a reparameterization of the two curves, encoding a synchronized motion along the two curves. That is, for a morphing $\mrp \in \MrpY{\cA}{\cB}$, and $t \in [0,\lenX{\mrp}]$, this encodes the configuration, with a point $\cA\bigl( x(\mrp(t)) \bigr) \in \cA$ matched with $\cB\bigl( y(\mrp(t)) \bigr) \in \cB$. The \emphi{elevation} of this configuration is
\begin{math}
    \elevC(t) =%
    \elevX{ \mrp(t) } =%
    \dY{ \cA\bigl( x(\mrp(t)\bigr) }{-\cB\bigl( y(\mrp(t)\bigr) }.
\end{math}

\subsubsection{\Frechet distance}

\begin{definition}%
    \deflab{width_f}%
    The \emphi{width} of a morphing $\mrp$ between $\cA$ and $\cB$ is
    \begin{math}
        \WidthX{\mrp} = \max_{ t \in [0,\lenX{\mrp}] } \elevC( t ).
    \end{math}
    The \emphi{\Frechet distance} between the two curves $\cA$ and $\cB$ is
    \begin{align*}
      \distFr{\cA}{\cB} = \min_{\mrp \in \MrpMY{\cA}{\cB}}
      \WidthX{\mrp},
    \end{align*}
    where $\MrpMY{\cA}{\cB} \subseteq \MrpY{\cA}{\cB}$ is the set of all $x/y$-monotone morphings.
\end{definition}

Conceptually, the \Frechet distance is the problem of computing the minimum bottleneck matching between two curves, respecting the order and continuity of the curves.\footnote{Formally, since the reparameterization is not one-to-one, this is not quite a matching. One can restrict to using only such bijections, with no adverse effects, but it adds a level of tediousness, which we avoid for the sake of simplicity of exposition.} Alternatively, it is an $L_\infty$-norm type measure of the similarity between two curves. It thus suffers from sensitivity to outliers. Furthermore, even if only a small portion of the morphing requires a long leash, the measure, and the algorithms computing it, may use this long leash in large portions of the walk, generating a matching that is loose in many places, see \figref{f_bad}.

Observe that the \Frechet distance is the minimum value such that the sublevel set of the elevation function has an $x/y$-monotone path from $(0,0)$ to $(\lenX{\curveA}, \lenX{\curveB})$ in $\doggy$.
\begin{figure}
    \phantom{}\hfill%
    \includegraphics[width=0.4\linewidth,height=3cm]{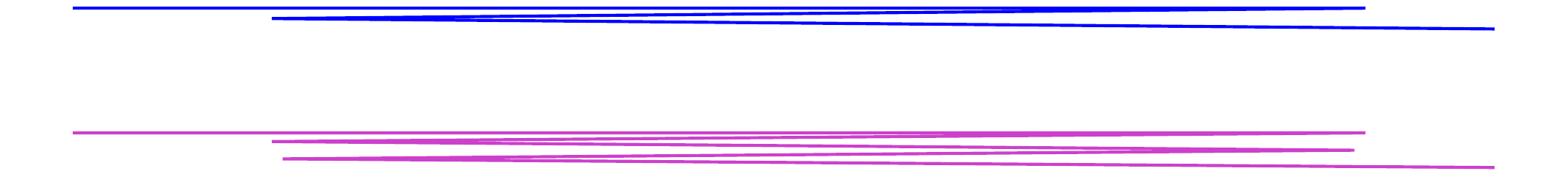} \hfill%
    \includegraphics[width=0.4\linewidth]{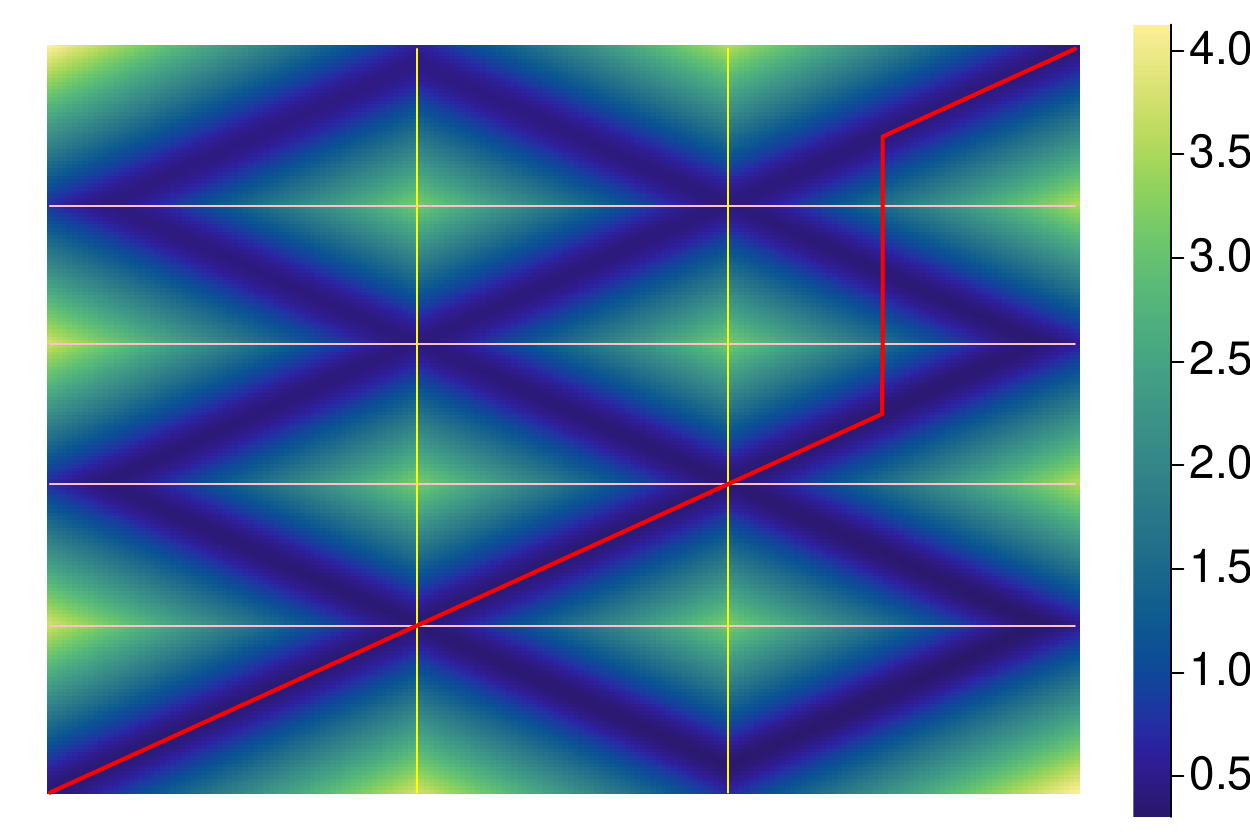}%
    \hfill%
    \phantom{}%
    \caption{Two curves, their free space diagram (and the associated elevation function), and the optimal \Frechet morphing between the two curves encoded as an $x/y$-monotone curve.  More illustrations and animations of this example are available \href{\baseUrlX{09}}{here}.  }%
    \figlab{Frechet_ex}
\end{figure}

\begin{figure}
    \centering
    \begin{minipage}{0.45\textwidth}
        \includegraphics[width=\linewidth,page=1]{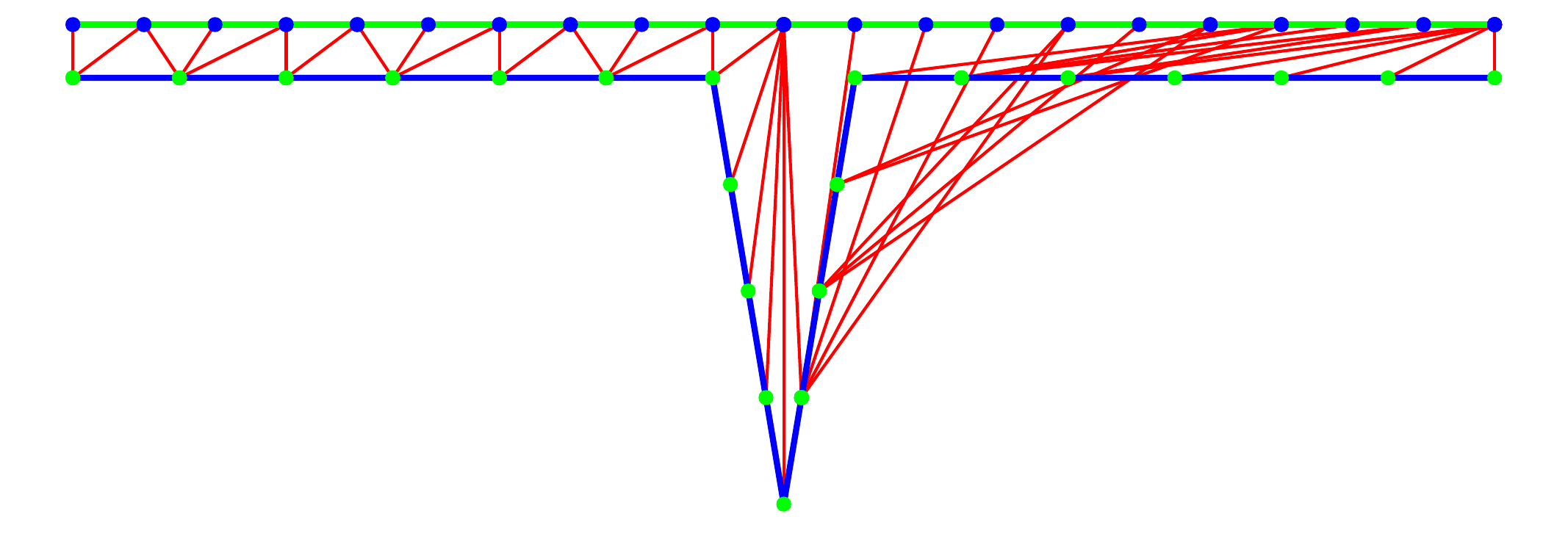}
        \begin{minipage}{0.9\textwidth}
            \caption*{The classical (discrete) \Frechet morphing, caring only about the maximum leash length.}
        \end{minipage}
    \end{minipage}
    ~
    \begin{minipage}{0.45\textwidth}
        \includegraphics[width=\linewidth,page=1]{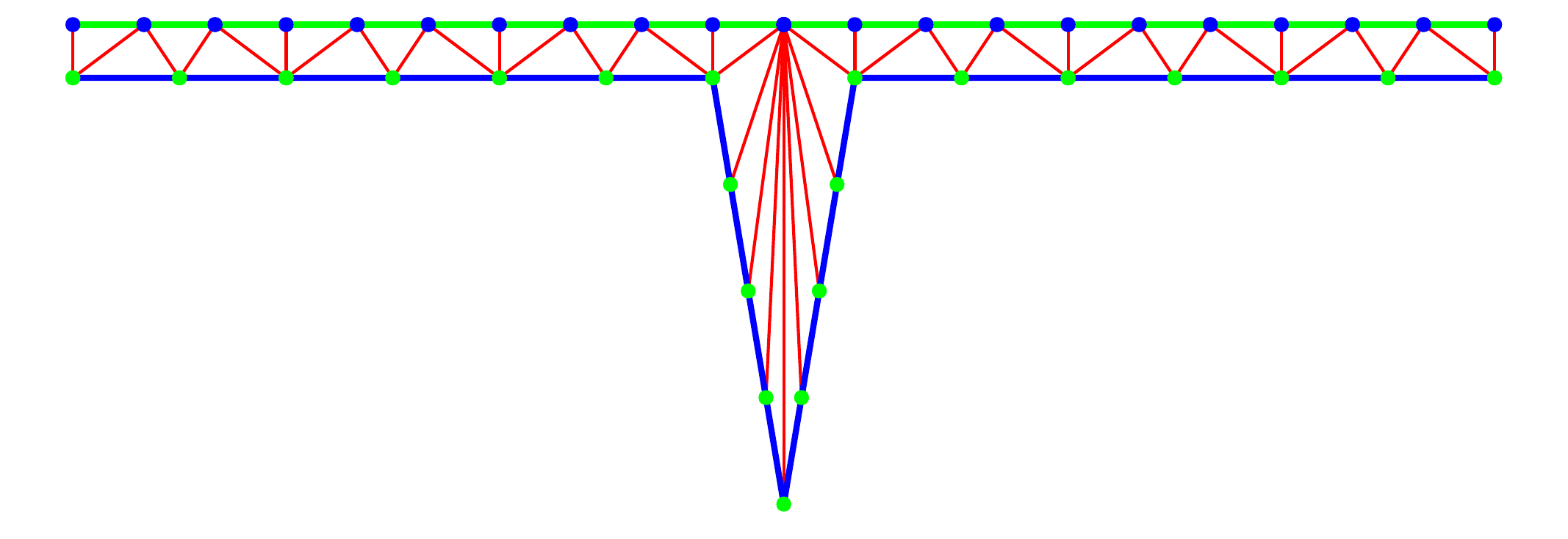}
        \begin{minipage}{0.9\textwidth}
            \caption*{The retractable discrete \Frechet morphing, using the shortest leash possible at each point.  }
        \end{minipage}
    \end{minipage}

    \caption{A comparison between the classical and retractable \Frechet{} distances. Observe that the morphing generated by the classical distance can be quite loose in many places. An animation of both morphings is available \href{\baseUrlX{01}}{here}.}
    \figlab{f_bad}
\end{figure}
\subsection{Background}

Alt and Godau \cite{ag-cfdbt-95} presented a rather involved $O(n^2 \log n)$ time algorithm to compute the \Frechet distance using parametric search. The parametric search can be removed by using randomization, giving a simpler algorithm as shown in \cite{hr-fdre-14}. Buchin \etal \cite{bblmm-cfdrl-16} presented an alternative algorithm for computing the \Frechet distance that replaces the decision procedure by using a data structure to maintain appropriate lower envelopes.  Despite some simplifications, all these algorithms are somewhat involved.

Unfortunately, it is believed this problem requires quadratic time in the worst case, although a logarithmic speedup is possible; see \cite{bbmm-fswdi-17} and references therein.  The quadratic time can be improved for realistic inputs by assuming that the input is ``nice'', and introducing approximation, but the resulting algorithms are still not simple \cite{dhw-afdrc-12}.

\subsubsection{Variants of the \Frechet distance}

\myparagraph{Weak \Frechet distance.}

The \emphi{weak \Frechet distance} allows morphings where the agents are allowed to go back and traverse portions of the curve visited previously (i.e., the morphing does not have to be $x/y$-monotone).  Since the weak \Frechet distance allows considering more parameterizations, it is potentially smaller, and we have that $\distWFr{\cA}{\cB} \leq \distFr{\cA}{\cB}$, for any two curves $\cA, \cB$.  Elegantly, Alt and Godau \cite{ag-cfdbt-95} showed that the weak \Frechet distance can be reduced to computing the minimum spanning tree of an appropriate graph. Unfortunately, there does not seem to be a natural way to overcome this non-monotonicity (and thus get the ``strong'' version).

\myparagraph{Discrete \Frechet distance.}

The complexity of these algorithms, together with the sensitivity of the \Frechet distance to noise, leads to using ``easier'' related measures, such as the discrete version of the problem, and dynamic time-warping (discussed below).  In the discrete version, you are given two sequences of points $p_1, \ldots, p_n$, and $q_1, \ldots, q_m$, and the purpose is for two ``frogs'' starting at $p_1$ and $q_1$, respectively, to jump through the points in the sequence until reaching $p_n, q_m$, respectively, while minimizing the maximum distance between the two frogs during this traversal. At each step, only one of the frogs can jump from its current location to the next point in its sequence (no jumping back). Computing the optimal distance under this measure can be done by dynamic programming, similar to the standard approach to edit distance.  Indeed, the configuration space here is the grid $\Grid = \IRX{n} \times \IRX{m}$, where $\IRX{n } = \{1,\ldots, n\}$.

To use the discrete version in the continuous case, one sprinkles enough points along both input curves and then solves the discrete version of the problem. Beyond the error this introduces, to get a distance that is close to the standard \Frechet distance, one has to sample the two curves quite densely in some cases.

For the (monotone) discrete \Frechet distance, the induced graph on the grid $\Grid$ is a \DAG, and the task of computing the \Frechet distance is to find a minimum bottleneck path from $(1,1)$ to $(n,m)$, where the weights are on the vertices.  Here, the weight on the vertex $(i,j)$ is the distance $\dY{p_i}{-q_j}$.  In particular, a \Frechet morphing is an $x/y$-monotone path in $\Grid$ from $(1,1)$ to $(n,m)$. The standard algorithm to do this traverses the grid, say, by increasing rows $i$, and in each row by increasing column $j$, such that the value at $(i,j)$ is the maximum of the length of the leash of this configuration, together with the minimum solution for $(i-1,j)$ and $(i, j-1)$. This algorithm leads to a straightforward, $O(n m)$ time algorithm for the discrete \Frechet distance. However, the morphing computed might be inferior, see \figref{f_bad} for such a bad example.

\myparagraph{Retractable \Frechet.}

For simplicity, assume that the pairwise distances between all pairs of points in the two sequences are unique. We would like to imagine that we have a retractable leash that can become shorter at times, and the leash ``resists'' being longer than necessary. It is thus natural to ask for a morphing where the leash is as short as possible at any point in time.

Informally, the optimal \emphw{retractable \Frechet morphing} between the two sequences includes the bottleneck configuration, realizing the \Frechet distance, in the middle of its path, and the two subpaths from the endpoints to this configuration have to be also recursively optimal. This concept is formally defined and described in \secref{r_d_f_dist}.  This concept was introduced by Buchin \etal \cite{bbms-lcfm-19}. Interestingly, they show that the discrete version can be computed in $O( nm )$ time, but unfortunately, the algorithm is quite complicated. They also show that the continuous retractable \Frechet can be computed in $O(n^3 \log n)$ time.

Buchin \etal \cite{bbms-lcfm-19} refers to this \Frechet distance as \emph{locally correct}, but we prefer the \emph{retractable} labeling.  The term ``retractable \Frechet'' was used by Buchin \etal \cite{bblmm-cfdrl-16}, but in a different (and not formally defined) context than ours.

\myparagraph{(Continuous) Dynamic Time Warping.}

One way to get less sensitivity for noise is to compute the total area ``swept'' by the leash as the walk is being performed. In the discrete case, we add up the lengths of the leashes during the configurations in the walk. There is also work on extending this to the continuous setting \cite{mp-cdtwt-99, bbknp-klmct-20}. For the continuous case, this intuitively boils down to computing (or approximating) an integral along the morphing.  See \secref{sweep_distance} for more details.

\subsubsection{Critical events}
\seclab{event_types}

The standard algorithm for computing the \Frechet distance works by performing a ``binary'' search for the \Frechet distance. Given a candidate distance, it constructs a ``parametric'' diagram that is a grid, where inside each grid cell the feasible region is a clipped ellipse. The task is then to decide if there is an $x/y$-monotone path from the bottom left corner to the top right corner, which is easily doable. The critical values the search is done over are: \medskip%
\begin{compactenumI}
    \item \emphw{Vertex-vertex events}: The distance between two vertices of the two curves,
    \item \emphw{Vertex-edge events}: The distance between a vertex of one curve and an edge of the other.

    \item \emphw{Monotonicity events}: This is the minimum distance between a point on one edge $\edge$ of the curves, and (maximum distance to) two vertices $u,v$ of the other curve. Specifically, it is realized by the point on $\edge$ with equal distance to $u$ and $v$.
\end{compactenumI}
\medskip%
The first two types of events are easy to handle, but the monotonicity events are the bane of the algorithms for the \Frechet distance.

\subsubsection{Algorithm engineering the \Frechet{} distance}

Given the asymptotic complexity and involved implementations of the aforementioned algorithms, there has been substantial work on practical aspects of computing the \Frechet{} distance.  In particular, in 2017, ACM SIGSPATIAL GIS Cup had a programming challenge to implement algorithms for computing the \Frechet distance. See \cite{wo-asgcr-2018} for details.

More recently, Bringmann \etal \cite{bkn-wdfpa-21} presented an optimized implementation of the decider for the \Frechet{} distance.  Somewhat informally, Bringmann \etal \cite{bkn-wdfpa-21} builds a decider for the \Frechet distance using a $kd$-tree over the free space diagram, keeping track of the reachable regions on the boundary of each cell, refining cells by continuing down the (virtual) $kd$-tree if needed.

\subsection{Our results}

\subsubsection{Result I: A new algorithm for retractable discrete \Frechet}

We observe that a natural approach to compute the retractable \Frechet morphing is to modify Dijkstra's/Prim's algorithm so that it solves the minimum bottleneck path problem. This observation leads to a simpler (but log factor slower) algorithm for computing it. The only modification of Dijkstra necessary is that one always handles the cheapest edge coming out of the current cut induced by the set of vertices already handled.  (In the discrete \Frechet case, the weights are on the vertices, but this is a minor issue.)  This modified version of Dijkstra is well known. Still, we include, for the sake of completeness, the proof showing that it indeed computes the recursively optimal path, which is also a retractable \Frechet morphing between the two sequences.  This algorithm results in better and more natural morphings, see \figref{f_bad}.

Maybe more importantly, in practice, one does not need to explore the whole space of $n m$ configurations (since we are in the discrete case, a configuration $(i,j) \in \IRX{n}\times\IRX{m}$ encodes the matching of $p_i$ with $q_j$), as the algorithm can stop as soon as it arrives at the destination configuration $(n,m)$.  Informally, if the discrete \Frechet distance is ``small" compared to the vast majority of pairwise distances (i.e., the two sequences are similar), then the algorithm only explores a small portion of the configuration space. Thus, this leads to an algorithm that is faster than the standard algorithm in many natural cases, while computing a significantly better output morphing.

\begin{figure}[t]
    \phantom{}%
    \hfill%
    \includegraphics[width=0.2\linewidth]{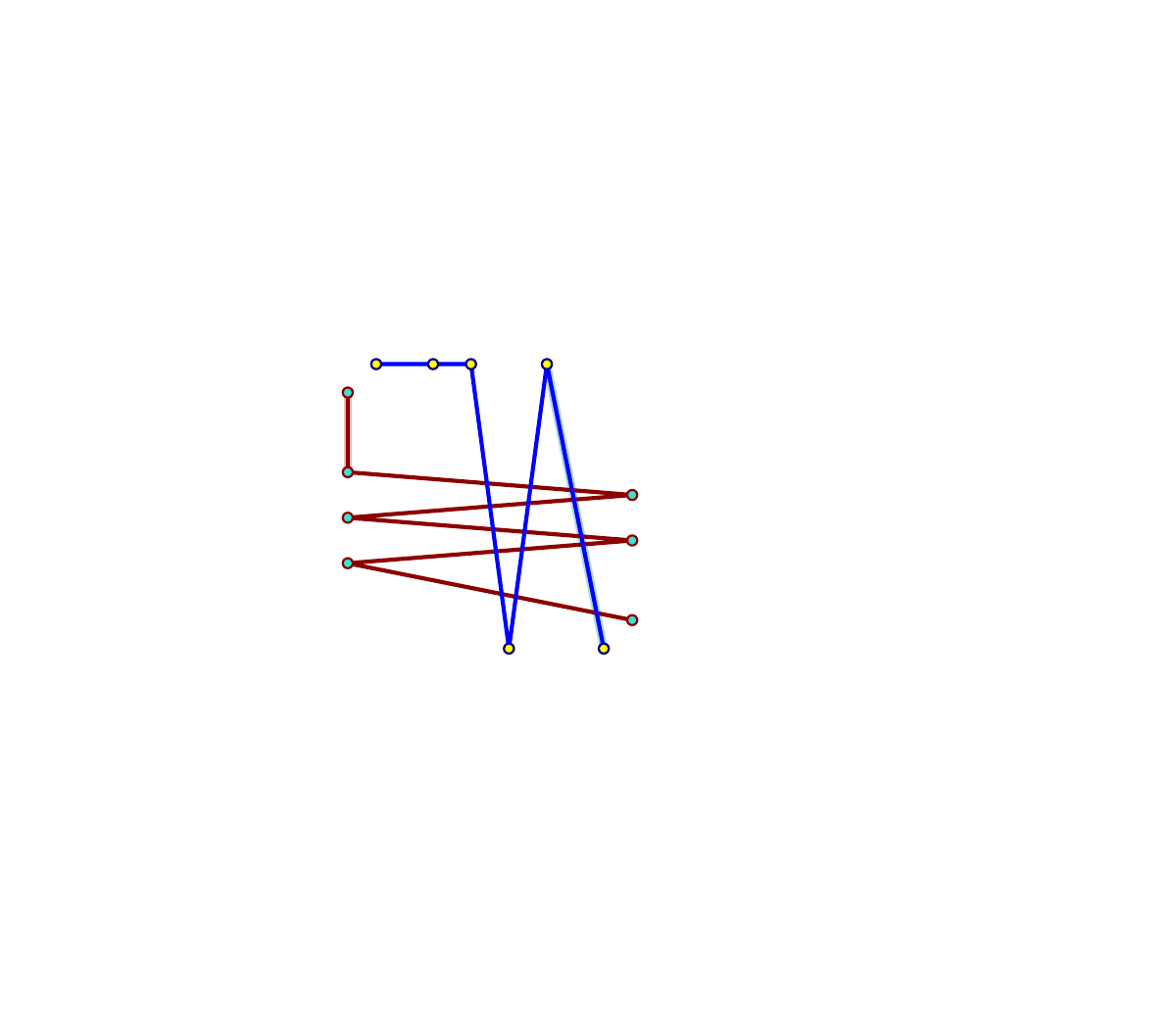}%
    \hfill%
    \includegraphics[width=0.3\linewidth]{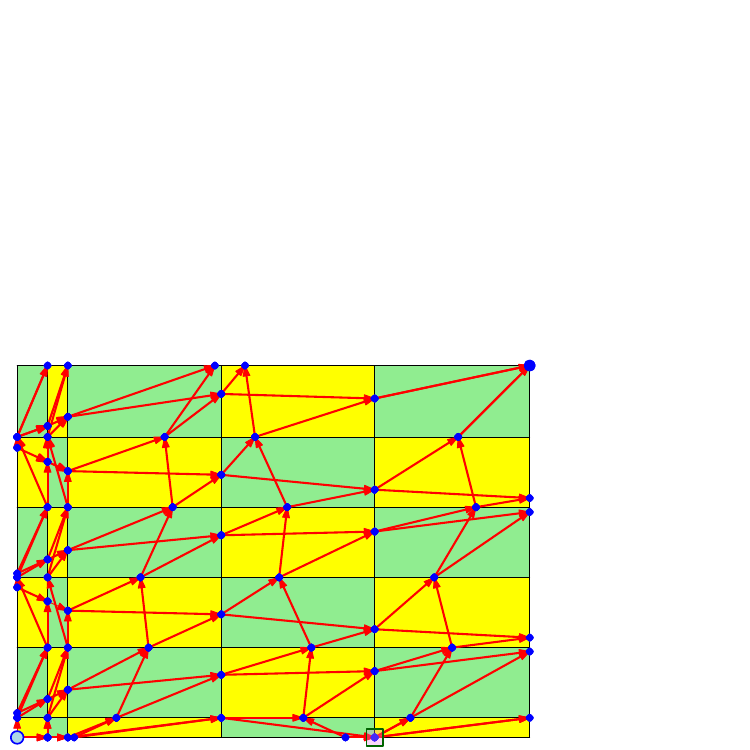}%
    \hfill%
    \includegraphics[width=0.38\linewidth]{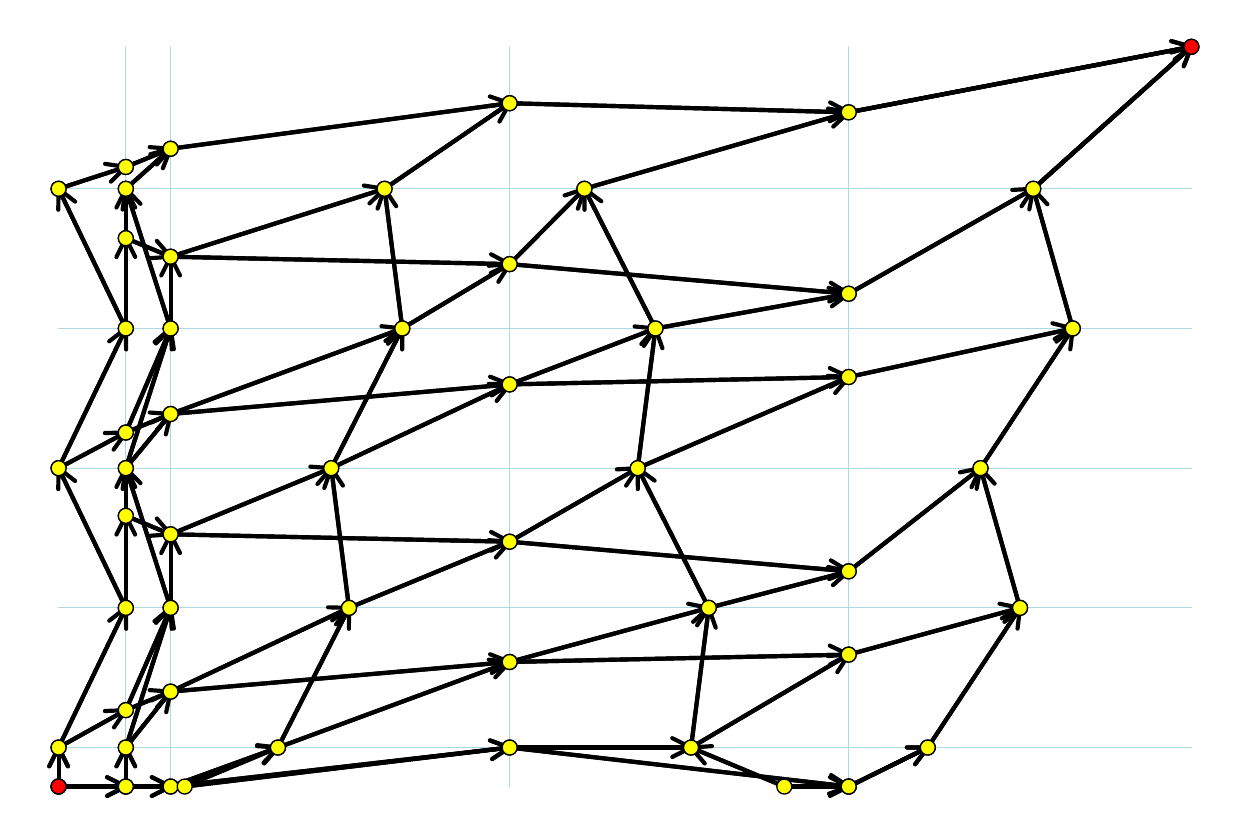} \hfill%
    \phantom{}%
    \caption{Two curves, and their associated \VEFrechet graph. %
       Note that every internal edge of the grid contains a portal (i.e., a vertex of the graph), but in many cases the portals of two edges are co-located on their common vertex (see the square-marked vertex in the middle figure). For the edges adjacent to the starting and ending corners, we set their portals to lie at the corners themselves.  See \href{\baseUrlX{10}}{here} for more info about a similar example.}
    \figlab{v_e_g_1}
\end{figure}

\subsubsection{Result II: A new distance and algorithm: \VEFrechet}

It is not clear how to extend the above to the continuous case.  A natural first step is to consider the continuous \Frechet distance, where one restricts the solution inside each cell of the free space diagram to be a segment (which is already the case for the morphings computed by existing algorithms), but more importantly, insisting that the shape of this segment must be determined only locally, thus facilitating a greedy strategy compatible with the retractable approach. In practical terms, we throw away the (global) monotonicity events.

\myparagraph{Traveling only through vertex-edge events.}

Because of the continuity and strict convexity of the elevation function, the function has a unique minimum on each edge of the free space diagram grid -- geometrically, this is the minimum distance between a vertex of one curve, and an edge of the other curve (a \emph{vertex-edge event}). We restrict our solution to enter and leave a cell only through these \emphw{portals} (which are easy to compute). The continuous configuration space now collapses to a discrete graph that is somewhat similar to the natural grid graph on $\IRX{n} \times \IRX{m}$. Indeed, a grid cell has four portals on its boundary edges. Specifically, there are directed edges from the portal on the bottom edge to the portal on the top and right edges of the cell. Similarly, there are edges from the portal on the left edge to the portals lying on the top and right edges. As before, the values are on the portals, and our purpose is to compute the optimal bottleneck path in this grid-like graph. Examples of this graph are depicted in \figref{v_e_g_1}, \figref{v_e_g_2} and \figref{v_e_g_3}.

\begin{remark}
    Munich and Perona \cite{mp-cdtwt-99} used a similar idea, but they used it in the other direction -- namely, in defining a better \CDTW distance for two discrete sequences. However, this idea was already present (implicitly) in the original work of Alt and Godau \cite{ag-cfdbt-95} -- indeed, their algorithm for the Weak \Frechet distance uses only the Vertex-Edge events (i.e., edges in the free space diagram).  This problem boils down to solving the bottleneck shortest path problem in an undirected graph.  In this case, this problem can be solved by computing the minimum spanning tree (e.g., by Prim's algorithm, which is a variant of Dijkstra's algorithm), as Alt and Godau do. For the directed case, one needs to use a variant of Dijkstra's algorithm \cite{gt-atbop-88} -- see \lemref{retractable}.  See also Buchin \etal \cite{bbdfj-ctmc-17} who also used a similar idea.
\end{remark}

We can now run the retractable bottleneck shortest-path algorithm (i.e., the variant of Dijkstra described above) on this implicitly defined graph, computing the vertices and edges of it, as they are being explored. For many natural inputs, this algorithm does not explore a significant fraction of the configuration space, as it involves distances that are significantly larger than the maximum leash length needed. The algorithm seems to have near-linear running time for many natural inputs.  The \emphi{\VEFrechet{}} morphing is the one induced by the computed path in this graph. Unfortunately, the \VEFrechet{} morphing might allow the agents to move backwards on an edge, but importantly, the motion across a vertex is monotone. Namely, the \VEFrechet is monotone for vertices, but not necessarily monotone on the edges. A vertex is thus a \emph{checkpoint} that once passed, cannot be crossed back.

\subsubsection{Result III: New algorithm for the regular \Frechet distance}

The natural question is how to use the (easily computable) \VEFrechet morphing to compute the optimal (regular) continuous \Frechet distance.  We next describe how this can be done in practice.

\myparagraph{The hunt for a monotone morphing.}  We denote the \VEFrechet distance between two curves $\cA$ and $\cB$ by $\distVEFr{\cA}{\cB}$. Clearly, we have that $\distWFr{\cA}{\cB} \leq \distVEFr{\cA}{\cB} \leq \distFr{\cA}{\cB}$.

One can, of course, turn any morphing into a monotone one by staying put instead of moving back. This approach is appealing for \VEFrechet, as the corresponding \VE morphing $\mrp$, say between two curves $\cA$ and $\cB$, never backtracks over vertices (only edges), so we already expect the error this introduces to be relatively small.  Let $\mrp^+$ denote the monotone morphing resulting from this simple strategy. Observe that
\begin{equation*}
    \WidthX{\mrp}
    =
    \distVEFr{\cA}{\cB}
    \leq
    \distFr{\cA}{\cB}
    \leq
    \WidthX{\smash{\mrp^+}}.
\end{equation*}
In particular, if $ \WidthX{\mrp} = \WidthX{\mrp^+} $, then $\distFr{\cA}{\cB}$ is realized by $\mrp^+$, and we have computed the \Frechet distance between $\cA$ and $\cB$.

A less aggressive approach is to introduce new vertices in the middle of the edges of $\cA$ and $\cB$ as to enforce monotonicity. Indeed, clearly, if we refine both curves by repeatedly introducing vertices into them, the \VEFrechet distance between the two curves converges to the \Frechet distance between the original curves, as introducing a vertex in the middle of an edge does not change the regular \Frechet distance, while preventing the \VEFrechet morphing from backtracking over this point.

We refer to this process of adding vertices to the two curves as \emphi{refinement}. (See \figref{f_refinement}.) In practice, in many cases, one or two rounds of (carefully implemented) refinement are enough to isolate the maximum leash in the morphing from the non-monotonicity, and followed by the above brute-force monotonization leads to the (practically) optimal \Frechet distance. Even for pathological examples, after a few more rounds of refinement, this process computes the almost-exact \Frechet distance. That is, the computed lower bound, which is the \VEFrechet distance, is equal to the width of the computed monotone morphing, which is the \Frechet distance.

\begin{remark}[Almost-exact: Floating point issues]
    \remlab{issues}%
    As we are implementing our algorithm using floating point arithmetic, and the calculation of the optimal \Frechet distance involves distances, imprecision is unavoidable. A slight improvement in precision can be achieved by using squared distances (and also slightly faster code) --- but for simplicity, we have not used this idea in our code. In particular, we take the somewhat pragmatic view that an approximation to the optimal up to a factor of (say) $1.00001$ can be considered as computing the ``optimal'' solution. We refer to such solutions as being \emphi{almost-exact}.

    Note that \Frechet morphings are somewhat less sensitive to numerical issues than other geometric problems --- indeed, once a morphing is computed, one can calculate its width directly.
\end{remark}

\begin{remark}[What if one wants the exact distance?]
    As pointed out above, our algorithm computes the almost-exact \Frechet distance, and in practice, there is no difference to the exact \Frechet distance (and in many cases they seem to coincide). Nevertheless, what if one insists on the exact \Frechet distance?

    The refinement process can be modified to compute the monotonicity events on the regions on the curves where monotonicity is being violated. This change would readily lead to an exact algorithm, though we did not pursue it any further.
    See the paper \cite{chr-cefdf-25}, who gave an exact implementation of our algorithm with some additional improvements.
\end{remark}

\subsubsection{Result IV: Computing the \Frechet distance quickly for real inputs}

The above leaves us with a natural strategy for computing the \Frechet distance between two given curves. Compute quickly, using simplification, a morphing between the two input curves, and maintain (using \VEFrechet, for example) both upper and lower bounds on the actual \Frechet distance. By carefully inspecting the morphing, (re)simplifying the curves in a way that is sensitive to their (local) \Frechet distance, and recomputing the above bounds, one can get an improved morphing. Repeat this process potentially several times till the upper and lower bounds meet, at which point the optimal \Frechet distance has been computed.

This approach seems somewhat overkill, but it enables us to compute (in practice) the almost-exact \Frechet distance between huge polygonal curves quickly.

\begin{figure}
    \centerline{\includegraphics[width=0.5\linewidth]{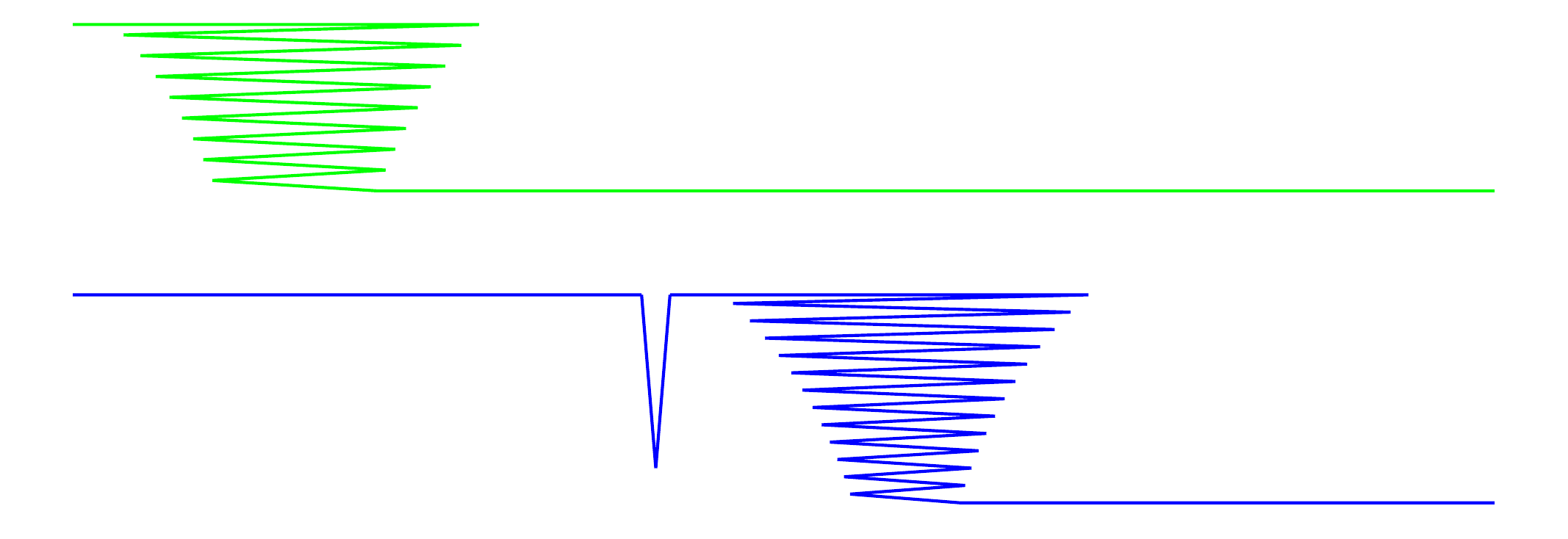}}
    \caption{For these two curves, the solution involves an agent stopping at one point on one curve while the other agent traverses a zig-zag, and vice versa. The algorithm enforces monotonicity by refining the two curves by introducing new vertices. For the results, see \href{\urlBaseX{06}}{here}.}
    \figlab{f_refinement}
\end{figure}

\subsubsection{Contribution: Implementation in \Julia and \Python}

We implemented the above algorithms as official packages in \Python and \Julia. A webpage with animations showing our algorithm in action is available \href{https://frechet.xyz/}{here}. The \Python package is available at \url{https://github.com/eliotwrobson/FrechetLib}, and the \Julia package is available at \url{https://github.com/sarielhp/FrechetDist.jl}.  Animations and examples computed by the new algorithm are available at \url{https://frechet.xyz/}.

\subsubsection{Additional results}

\myparagraph{Sweep distance.}

We demonstrate how one can convert our algorithm for computing \VEFrechet to an algorithm that computes a variant of the \CDTW distance, which we call the \emphw{sweep distance}. One can then use refinement to approximate the \CDTW distance. One can also compute a lower bound on this quantity, and the two quantities converge.  See \secref{sweep_distance} for details.

\myparagraph{Fast simplification.}

We show how to preprocess a curve $\cA$ with $n$ vertices, in $O(n \log n)$ time, such that given a query $w$, one can quickly extract a simplification of $\cAs$ of $\cA$, such that $\distFr{\cA}{\cAs} \leq w$. Importantly, the time to extract $\cAs$ is proportional to its size (i.e., the extraction is output sensitive). More importantly, in practice, this works quite well -- the size of $\cardin{\VX{\cAs}}$ is reasonably close (by a constant factor) to the optimal approximation. Combined with greedy simplification, this yields a very good simplification result. See \secref{fast_simplifier} for details.

\subsection{Significance}

We demonstrate in this paper that a minor variant of the \textbf{\emph{continuous}} \Frechet distance can be computed by a straightforward Dijkstra-type algorithm. Furthermore, for many real-world inputs, it runs in near-linear time (out of the box). Similarly, for many real-world inputs, the computed morphing is monotone, thus realizing the (standard) \Frechet distance. More importantly, because of the \emph{retractable} nature of the morphing computed, the matching computed is more natural, and can handle noise/outliers more gracefully than the matchings calculated by the current algorithms for the \Frechet distance. Indeed, areas that are noise/outliers get isolated in the morphing to a small interval, which can be easily identified and handled, see \figref{f_bad}. We believe that this makes the morphings computed by our new algorithms significantly better for real-world applications than previous algorithms.  There is also previous work on other variants of \Frechet distance that are more robust to uncertainty and outliers \cite{fnpr-fed-24}, and shortcutting \cite{dh-jydfd-13}.

We then show how to modify this algorithm to compute the (monotone) \Frechet distance (in cases when the morphing was not already monotone), and how to make it handle large inputs quickly via simplifications (handling all the technical difficulties this gives rise to).

Finally, we implemented our new algorithms in \Julia and \Python. We made them publicly available as standard packages, thus making the computation of the \Frechet distance accessible to a broader audience.  Using such packages in \Python/\Julia is significantly easier than using any non-pre-compiled code in \CPP.

\smallskip

The combination of simplification, new (and not so new) algorithmic ideas (such as retractability via Dijkstra, and \VEFrechet, among others), and careful implementation, leads to reasonably fast performance in practice.

\subsection{Errata}
\seclab{err}

An earlier version of our paper \cite{hrr-fdu-25} made the extremely incorrect claim that our \Julia implementation is faster than the \CPP implementation of Bringmann \etal \cite{bkn-wdfpa-21}. This error was due to both misconfiguration of the data and incorrect measurement of the running times of the \CPP implementation. In reality, the \CPP algorithm is, by several orders of magnitude, faster than our implementation/algorithm. We believe that the ideas in our paper may still be of interest. See \cite{chr-cefdf-25}, who recently turned our approach into an exact algorithm after making some additional improvements, and provided an empirical evaluation.

\section{Implementation and experiments}

To demonstrate the practicality of our techniques, we have developed open-source packages that implement our algorithms. Our implementations follow the ideas described in the paper -- aggressively using simplification, and upper/lower bound computations to guarantee we computed the optimal \Frechet distance.  Furthermore, because our algorithms use the retractable \VEFrechet as our building block, the \Frechet morphings our implementation computes look better than what the standard \Frechet algorithms would return.

To see numerous animations and more information, the reader is encouraged to visit the (anonymized) site: \url{https://frechet.xyz}.

\subsection{Hardware and implementations}

The tests were all performed on a Linux system with 64GB of memory and an Intel \texttt{i7-11700} CPU with 16 threads. This system is a fast desktop, but far from the best possible hardware currently available. We were also able to perform development for the experiments on a laptop with substantially weaker hardware, giving anecdotal evidence of the scalability of our algorithms.  We use several real-world datasets featured in prior works to demonstrate the effectiveness of our algorithm.

Here, we discuss the individual implementations used in our experiments. The first two (\Julia and \Python) use the algorithmic techniques discussed in this work. We emphasize, however, that our implementation is slower by several orders of magnitude than the implementation of Bringmann \etal \cite{bkn-wdfpa-21}.

\subsubsection{Julia}

The \Julia code is about 5000 lines for the library, and an additional 5000 for the examples/testing code. It implements the algorithms described here faithfully, with some additional hacks to handle floating-point issues. \Julia has (surprisingly easy to use) support for multi-threading (unlike \texttt{C++} or \Python), and we also tested a multi-threaded version.  The code is available from \url{https://frechet.xyz/}, including instructions on how to replicate our tests.

The \Julia package is available at \url{https://github.com/sarielhp/FrechetDist.jl}.  Animations and examples computed by the \Julia code are available at \url{https://frechet.xyz/}.

\myparagraph{\Julia low-level optimizations.}

We were able to significantly improve the performance of the \Julia code using some standard low-level optimizations: (i) \Julia performs inlining code automatically, but explicitly forcing it to inline some low-level key functions (such as computing the distance between two points), improved the performance measurably ($>10\%$). \Julia also has high overhead for creating \texttt{s{}tr{}u{c}t{}s} that get destroyed immediately (essentially because it allocates them on the heap, instead of the stack like \CPP). Thus, by rewriting the code so that computing the distance of a point from a segment is done directly on the endpoints of the segment (where the segment is an edge of a polygon), performance was also improved. Caching simplifications, so that one can extract/compute new simplification quickly, among other ideas, see \secref{fast_simplifier}, helped create fast code.

\myparagraph{One allocation to rule them all.}  Since \Julia performance suffers when a lot of dynamic allocations are done (because of the automatic garbage collection), it turned out that implementing the retractable \VEFrechet without using hashing (for the vertices of the graph) resulted in much faster code. To this end, one pre-allocates a quadratic-sized array representing the vertices of the grid (but exploring only the relevant portions of the grid/this array) --- this works because the allocation/initialization of this (single!)  array is high-speed. This seems to yield significant speedups, somewhere up to ten times faster than the hashing version, depending on the input (for large inputs, the two implementations seem to have similar running times).

\newcommand{\Struct}{\texttt{struct}\xspace}%

\paragraph*{Julia mutable vs immutable behavior.}
Another interesting technicality is the way \Julia stores \Struct{}s in arrays. If the \Struct is declared mutable (i.e., its fields can be changed), an array of such \Struct{}s is implemented as an array of pointers to a singleton \Struct, each individually allocated on the memory heap. If \Struct{} is immutable, which is the default, such arrays are contiguous \Struct{}s placed one after the other in memory (as in \CPP).  Somewhat counterintuitively, an immutable \Struct{} can be modified by copying a newly constructed \Struct into it.  This behavior gives rise to differences in performance. It is motivated by the idea that, say, points rarely require a single coordinate to be changed, and once created, they tend not to be modified.  In any case, we changed our point type to be immutable, and it led to significant improvements in running time.

\paragraph*{Parameter tuning.}

A significant amount of speedup resulted from optimizing various parameters in the algorithm implementation. One can usually guess what should be roughly the right range for these parameters, but in the end, we manually searched over some values, which was both tedious and heuristic. The problem of parameter tuning is, of course, well known and widely studied in machine learning -- and using such techniques might be the next stage in further optimizing the implementation.

\subsubsection{Python}

For the \Python implementation, we represented curves using NumPy arrays \cite{numpy-20}, where each row in the array is a single point on the input curve. This representation is memory efficient and allows for fast computations of distances between points using library primitives. We also made heavy use of Numba \cite{lps-nlpjc-15} to further improve performance, as many of the algorithms discussed in this work are iterative and thus easy to accelerate using the library.

\myparagraph{Numerical Issues.} %
Although not exclusive to this implementation, the \Python code suffers from some issues related to numerical precision. Specifically, this implementation struggles with underflow of floating-point numbers, which seems to occur in situations where consecutive points are too close in the input data (i.e., the distance between points is too small). However, this doesn't represent a meaningful distance between these points, and this can be resolved by adding simplification as a preprocessing step (such as the algorithm in \secref{fast_simplifier}).

Our experiments are meant to emphasize the performance of our algorithmic techniques, and we do not have any tasks specifically designed to evaluate issues with floating-point numbers or methods that can be used to combat this problem. In particular, we did not encounter these issues to any significant extent with the datasets mentioned in this work. This is a potential avenue for future work if given data that encounters these issues frequently.

\subsubsection{\CPP}
The \CPP implementation is provided by Bringmann \etal \cite{bkn-wdfpa-21}, and it is faster by several orders of magnitude than our implementation.

\subsection{Experiments}

To compare our implementations, we performed benchmarks using two different tasks using real-world datasets. These were selected to determine the scalability of the algorithms relative to the large size of real-world data. This section is organized by the different tasks performed.

\subsubsection{Direct \Frechet distance computations}

To compare the \Julia and \Python implementations of our algorithms, we directly computed the \Frechet distance at different levels of approximation using our algorithms.

This was done on individual pairs of curves selected from different datasets. Specifically, we use curves taken from the GeoLife dataset \cite{zfxml-ggtdu-11} (as used in \cite{bkn-wdfpa-21}), a stork migration dataset \cite{rktzw-eabgc-18}, (retrieved from Movebank \cite{rktzw-deabg-18}, used in \cite{bbknp-klmct-20}), and a pigeon flight dataset \cite{pgrbvg-pigeon-17}. \figref{datasets} lists the individual curves from each dataset used in our benchmarks.

We do not perform any transformations on the input data (to reduce aberration due to Earth's curvature, for example), as we are primarily interested in the execution speed of our algorithm.

See \figref{python_results} and \figref{Julia_vs_python} for the results (observe that the \Julia code was tested for $1.001$ approximation -- we have not done the same for the \Python code, since it was numerically unstable in this case).  The \Python performance on these datasets is competitive with the \Julia implementation, although in some cases significantly slower. In general, the \Julia implementation is faster (we did not perform enough experiments to decide how much faster). Still, the \Python implementation is fast and robust enough to be used in practice. The runtimes also do not include the Numba JIT compilation times.

Specifically, the slow performance observed on some test cases for the \Python implementation may be caused by using sub-optimal data structures for some simplification tasks. We hope to improve this performance in further development of this implementation.

\subsection{Discussion}

Overall, our \Julia implementation is still significantly slower, by several orders of magnitude, than the \CPP of Bringmann \etal \cite{bkn-wdfpa-21}. This leaves it open whether any of the ideas used in this paper can be used together with their algorithm to get a faster algorithm.

One optimization used by the \Julia code, that should be generally useful, is the following. As a preprocessing step, precompute a hierarchy of simplified curves for each input curve. This improves the query process, but makes the preprocessing more expensive. Thus, storing such precomputed hierarchies might be a good idea if input curves are going to be used repeatedly for performing distance queries.  We emphasize that the reported running times include this (light) preprocessing stage (interestingly, the \si{SIGSPTIAL} \cite{wo-asgcr-2018} competition allowed such preprocessing to not be included in the overall running time).

\paragraph*{Acknowledgements.}
We thank the authors of \cite{chr-cefdf-25} and \cite{bkn-wdfpa-21}
for observing that our prior comparisons with the \CPP implementation
of Bringmann \etal \cite{bkn-wdfpa-21} were incorrect, as discussed in
\secref{err}.

\begin{figure}[h!]
    \centering%

\begin{tabular}{|c|l|}
  \hline
  \textbf{Input} & \textbf{Description} \\\hline
  \cellcolor{lightgray}{\texttt{1}} & \cellcolor{lightgray}{\texttt{birds: 1787\_1 / 1797\_1}} \\
  \texttt{2} & \texttt{birds: 2307\_3 / 2859\_3} \\
  \cellcolor{lightgray}{\texttt{3}} & \cellcolor{lightgray}{\texttt{birds: 2322\_2 / 1793\_4}} \\
  \texttt{4} & \texttt{GeoLife 20080928160000 / 20081219114010} \\
  \cellcolor{lightgray}{\texttt{5}} & \cellcolor{lightgray}{\texttt{GeoLife  20090708221430 / 20090712044248}} \\
  \texttt{6} & \texttt{Pigeons RH887\_1 / RH887\_11} \\
  \cellcolor{lightgray}{\texttt{7}} & \cellcolor{lightgray}{\texttt{Pigeons C369\_5 / C873\_6}} \\
  \texttt{8} & \texttt{Pigeons C360\_10 / C480\_9} \\\hline
\end{tabular}

    \caption{The inputs tested and where they are taken from, \texttt{birds} referring to the stork migration dataset \cite{rktzw-deabg-18}, the GeoLife dataset \cite{zfxml-ggtdu-11}, and \texttt{Pigeons} referring to a dataset from \cite{pgrbvg-pigeon-17}.}
    \figlab{datasets}
\end{figure}

\begin{figure}[h!]
    \centering%
\begin{tabular}{|c|rr|rrrrr|r|}
  \hline
  \textbf{Input} & \textbf{P \#} & \textbf{Q \#} & \textbf{$\approx$4} & \textbf{$\approx$1.1} & \textbf{$\approx$1.01} & \textbf{$\approx$1.001} & \textbf{Exact} & \textbf{VER} \\\hline
  \cellcolor{lightgray}{\texttt{1}} & \cellcolor{lightgray}{\texttt{10,406}} & \cellcolor{lightgray}{\texttt{11,821}} & \cellcolor{lightgray}{\texttt{0.106}} & \cellcolor{lightgray}{\texttt{0.256}} & \cellcolor{lightgray}{\texttt{4.164}} & \cellcolor{lightgray}{\texttt{16.639}} & \cellcolor{lightgray}{\texttt{0.476}} & \cellcolor{lightgray}{\texttt{27.743}} \\
  \texttt{2} & \texttt{16,324} & \texttt{14,725} & \texttt{0.114} & \texttt{0.428} & \texttt{0.985} & \texttt{4.762} & \texttt{3.522} & \texttt{19.734} \\
  \cellcolor{lightgray}{\texttt{3}} & \cellcolor{lightgray}{\texttt{22,316}} & \cellcolor{lightgray}{\texttt{4,613}} & \cellcolor{lightgray}{\texttt{0.151}} & \cellcolor{lightgray}{\texttt{0.651}} & \cellcolor{lightgray}{\texttt{1.207}} & \cellcolor{lightgray}{\texttt{3.262}} & \cellcolor{lightgray}{\texttt{0.882}} & \cellcolor{lightgray}{\texttt{7.138}} \\
  \texttt{4} & \texttt{56,730} & \texttt{91,743} & \texttt{0.578} & \texttt{1.207} & \texttt{4.215} & \texttt{28.109} & \texttt{4.259} & \texttt{---} \\
  \cellcolor{lightgray}{\texttt{5}} & \cellcolor{lightgray}{\texttt{6,103}} & \cellcolor{lightgray}{\texttt{9,593}} & \cellcolor{lightgray}{\texttt{0.043}} & \cellcolor{lightgray}{\texttt{0.040}} & \cellcolor{lightgray}{\texttt{0.305}} & \cellcolor{lightgray}{\texttt{0.519}} & \cellcolor{lightgray}{\texttt{0.484}} & \cellcolor{lightgray}{\texttt{5.678}} \\
  \texttt{6} & \texttt{2,702} & \texttt{1,473} & \texttt{0.034} & \texttt{0.084} & \texttt{1.367} & \texttt{1.387} & \texttt{0.160} & \texttt{1.081} \\
  \cellcolor{lightgray}{\texttt{7}} & \cellcolor{lightgray}{\texttt{1,068}} & \cellcolor{lightgray}{\texttt{1,071}} & \cellcolor{lightgray}{\texttt{0.026}} & \cellcolor{lightgray}{\texttt{0.066}} & \cellcolor{lightgray}{\texttt{0.048}} & \cellcolor{lightgray}{\texttt{0.049}} & \cellcolor{lightgray}{\texttt{0.044}} & \cellcolor{lightgray}{\texttt{0.011}} \\
  \texttt{8} & \texttt{864} & \texttt{1,168} & \texttt{0.018} & \texttt{0.044} & \texttt{0.142} & \texttt{0.124} & \texttt{0.069} & \texttt{0.080} \\\hline
\end{tabular}

    \caption{\Julia performance on various inputs. All running times are in seconds. The $\#$ columns specify the number of vertices in each input. The single missing running time is for a case where the program ran out of memory. The quality of approximation is in the title of the column. The \text{V{E}R} column is for running the \VEFrechet (retractable) algorithm on the original input curves, which is much slower than the exact algorithm, which computes the exact \Frechet distance, but uses simplification internally for speed.}
    \figlab{Julia_vs_python}
\end{figure}

\begin{figure}[h!]
    \centering%
\begin{tabular}{|c|rr|rrrrr|}
  \hline
  \textbf{Input} & \textbf{P \#} & \textbf{Q \#} & \textbf{$\approx$4} & \textbf{$\approx$1.1} & \textbf{$\approx$1.01} & \textbf{Exact} & \textbf{VER} \\\hline
  \cellcolor{lightgray}{\texttt{1}} & \cellcolor{lightgray}{\texttt{10,400}} & \cellcolor{lightgray}{\texttt{11,815}} & \cellcolor{lightgray}{\texttt{0.071}} & \cellcolor{lightgray}{\texttt{0.250}} & \cellcolor{lightgray}{\texttt{5.372}} & \cellcolor{lightgray}{\texttt{0.593}} & \cellcolor{lightgray}{\texttt{92.869}} \\
  \texttt{2} & \texttt{16,318} & \texttt{14,719} & \texttt{0.091} & \texttt{0.240} & \texttt{1.142} & \texttt{69.525} & \texttt{66.844} \\
  \cellcolor{lightgray}{\texttt{3}} & \cellcolor{lightgray}{\texttt{22,310}} & \cellcolor{lightgray}{\texttt{4,607}} & \cellcolor{lightgray}{\texttt{0.087}} & \cellcolor{lightgray}{\texttt{0.232}} & \cellcolor{lightgray}{\texttt{1.041}} & \cellcolor{lightgray}{\texttt{49.426}} & \cellcolor{lightgray}{\texttt{25.032}} \\
  \texttt{4} & \texttt{56,730} & \texttt{91,743} & \texttt{0.307} & \texttt{0.563} & \texttt{1.892} & \texttt{2.390} & \texttt{---} \\
  \cellcolor{lightgray}{\texttt{5}} & \cellcolor{lightgray}{\texttt{6,103}} & \cellcolor{lightgray}{\texttt{9,593}} & \cellcolor{lightgray}{\texttt{0.042}} & \cellcolor{lightgray}{\texttt{0.032}} & \cellcolor{lightgray}{\texttt{0.175}} & \cellcolor{lightgray}{\texttt{3.657}} & \cellcolor{lightgray}{\texttt{19.618}} \\
  \texttt{6} & \texttt{2,696} & \texttt{1,467} & \texttt{0.077} & \texttt{0.294} & \texttt{4.913} & \texttt{2.946} & \texttt{283.538} \\
  \cellcolor{lightgray}{\texttt{7}} & \cellcolor{lightgray}{\texttt{1,062}} & \cellcolor{lightgray}{\texttt{1,065}} & \cellcolor{lightgray}{\texttt{0.145}} & \cellcolor{lightgray}{\texttt{0.275}} & \cellcolor{lightgray}{\texttt{0.455}} & \cellcolor{lightgray}{\texttt{0.666}} & \cellcolor{lightgray}{\texttt{382.455}} \\
  \texttt{8} & \texttt{858} & \texttt{1,162} & \texttt{0.034} & \texttt{0.100} & \texttt{1.537} & \texttt{0.159} & \texttt{0.709} \\\hline
\end{tabular}

    \caption{Python performance in seconds on the inputs given in \figref{datasets}. The missing runtime is for a case where the test code ran out of memory.}
    \figlab{python_results}
\end{figure}

\section{The retractable \Frechet distance}

\subsection{The retractable path in a directed graph}

Let $\G=(\VV,\EE)$ be a directed graph with $n$ vertices and $m$ edges, with weights $\wC:\EE \rightarrow \Re$ on the edges (assume for simplicity of exposition that they are all distinct).  Consider a simple path $\pi$ in $\G$. The \emphi{bottleneck} of $\pi$ is $b(\pi) = \max_{e \in \pi} \wX{ e}$. For any two vertices $s$ and $t$ in $\G$, let $\Pi(s,t)$ denote the set of all simple paths from $s$ to $t$ in $\G$.  The \emphw{bottleneck distance} between $s$ and $t$ is $d_B(s,t) = \min_{\pi \in \Pi(s,t)} b(\pi)$. The \emph{unique} edge (under our assumption of distinct weights on the edges) that realizes $d_B(s, T)$ is the \emphi{bottleneck edge} of $s$ and $t$, denoted by $\btl(s,t)$.

\myparagraph{Literature on computing the bottleneck edge/path.}

This variant of the problem is a $\min \max$, which is equivalent to the $\max \min$ (maximize the cheapest edge on the path) via negation of the prices of edges. If the edges are sorted, then computing the bottleneck edge can be done in linear time. Dijkstra can be modified to solve this problem \cite{cktzz-bptdg-16}, in $O(n \log n + m)$ time. Similarly, the problem is readily solvable in linear time if the underlying graph is a \DAG.  Gabow and Tarjan \cite{gt-atbop-88} showed an algorithm with running time $O( m \log^* m)$, and some minor further improvements are known \cite{cktzz-bptdg-16}.

\myparagraph{A path that is bottleneck-optimal for subpaths.}

Here, we study the following harder variant of the bottleneck problem.

\begin{definition}
    \deflab{retract}%
    For two vertices $s,t$ of a directed graph $\G$ with distinct weights on its edges, the \emphi{retractable} path from $s$ to $t$, is the unique path that contains the edge $\btl(s,t) = u\rightarrow v$, and the subpath from $s$ to $u$, and from $v$ to $t$ are both retractable.
\end{definition}
Importantly, we are interested in computing the retractable path itself (not just the bottleneck). Confusingly, and this is not obvious, the standard dynamic programming algorithm for \DAG{}s does not work in this case.

\myparagraph{Modifying Dijkstra.}

It is not hard to see that a minor simplification to Dijkstra's algorithm, started from $s$, leads to a $O(n \log n + m)$ time algorithm for computing the retractable tree, containing the retractable paths from $s$ to all the vertices in $\G$ (we assume here that all of $\G$ is reachable from $s$). Indeed, if $C$ is in the current set of vertices already visited, the algorithm always handles the next cheapest edge in the directed cut $(C, \VV \setminus C)$. As in Dijkstra's algorithm, one can maintain a variable $d(v)$ to store the weight of the cheapest edge from a vertex already visited to the (yet unvisited) vertex $v$. Setting $d(v) = \min( d(v), \wX{ u\rightarrow v} )$ when handling the edge $u \rightarrow v$, records this information. Under this vertex-based accounting, the algorithm always handles the next vertex with the minimum $d$ value that is not yet visited (and these values can be maintained in a heap).  Gabow and Tarjan \cite{gt-atbop-88} described how to modify Dijkstra to solve the bottleneck problem -- what we point out in the next lemma is that it is somewhat stronger -- it computes the retractable path.

\begin{lemma}
    \lemlab{retractable}%
    The above algorithm computes a directed tree $T$ rooted at $s$, such that for any vertex $u$, the path in $T$ from $s$ to $u$ is retractable. The running time of the algorithm is $O(n \log n + m )$.
\end{lemma}
\begin{proof}
    The running time follows as the algorithm performs $O(m)$ decrease-keys and $O(n)$ delete-min operations, which take $O(1)$ and $O(\log n)$ time, respectively, if using a Fibonacci heap.

    As for correctness -- we analyze a somewhat slower implementation of the algorithm (i.e., this is a variant of Prim's algorithm using cuts). The algorithm repeatedly handles the cheapest edge not handled yet that is in the directed cut $(C, \overline{C})$\footnote{$(C, \overline{C}) = \Set{u\rightarrow v \in \EGX{\G}}{ u \in C, v \in \VX{G} \setminus C}$.}, where $C$ is the set of vertices visited so far. Initially, $C = \{s \}$ and the min-heap is initialized to hold the set $(C, \overline{C})$.  The algorithm now repeatedly extracts the minimum edge $u \rightarrow v$ in the heap. If $v \in C$, then the algorithm continues to the next iteration. Otherwise, the algorithm marks the $v$ as visited, and adds all the outgoing edges from $v$ to the heap. It is easy to verify (by induction) that this slower algorithm and the algorithm described above compute the same bottleneck tree.

    As for the correctness of the slower algorithm -- the proof is by induction on the number of edges of the retractable path.  Consider the case that the retractable path is of length one -- namely, the edge $e = s \rightarrow u$. Let $\wC = \wX{ e}$. Clearly, before the algorithm handles $e$, all the weights being handled are strictly smaller than $e$, and thus the algorithm can not arrive at $u$. Since the algorithm schedules $e$, the claim readily follows.

    So, consider a retractable path $\pi$ from $s$ to $u$ that contains more than one edge, and let $e = x \rightarrow y$ be its bottleneck edge. Let $\wC = \wX{e}$. Observe that the algorithm would handle all the edges of $\pi$ before handling any edge with strictly larger weight than $w$. As such, for our purposes, assume that $e$ is the heaviest edge in the graph. Let $S$ be all the vertices visited just before $e$ was handled, and let $T = \VV \setminus S$. Consider the two induced graphs $\G_S$ and $\G_T$, and observe that one can apply induction to both parts, separately, as there are no edges except $e$ between the two parts (there might be edges going ``back'', but these can be ignored, as the algorithm already handled the vertices of $S$). It thus follows that the algorithm computed $\pi$ as the desired path.
\end{proof}

\begin{figure}[t]
    \centering%
    \centering \includegraphics[width=0.96\linewidth]{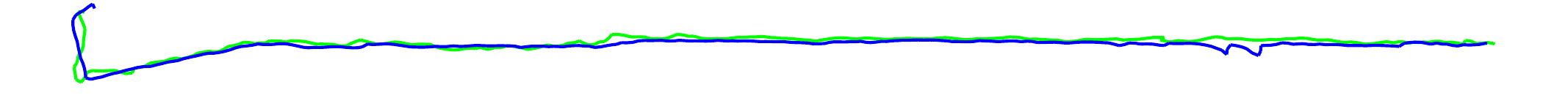} \vspace{-0.3cm}
    \caption{Close encounters of the \Frechet type: An example of two curves (from Geolife GPS tracks) that are made out of $547$ vertices, such that their retractable discrete \Frechet distance was computed by visiting only $1,144$ edges, while the whole diagram has $74,236$ distinct cells. For more details, follow this \href{\baseUrlX{12/}}{link}.  Informally, for examples where the \Frechet distance is dramatically smaller than the diameter of the curves, the retractable \Frechet distance (discrete or continuous) is computed by the new algorithm in near linear time.}%
    \figlab{near_linear}
\end{figure}

\subsection{The retractable discrete \Frechet distance}
\seclab{r_d_f_dist}

\subsubsection{The discrete \Frechet distance}

Let $\curveA = p_1, \ldots, p_n$ and $\curveB = q_1, \ldots, q_m$ be two sequences of points in some normed space. Conceptually, we have two agents starting at $p_1$ and $q_1$ respectively. During an atomic time interval, one of them can jump forward to the next point in the sequence (one can allow both to jump forward in the same time, but we disallow it for the sake of simplicity of exposition). In the end, both agents have to arrive at $p_n$ and $q_m$, respectively, and our purpose is to minimize the maximum distance between them during this motion.

\begin{definition}
    Consider two sequences $\curveA = p_1, \ldots, p_n$ and $\curveB = q_1, \ldots, q_m$, a \emphi{discrete morphing} is an $x/y$-monotone path $\curveC = (1,1), \ldots, (n,m)$ in the grid graph defined over the set of points $U = \IRX{n}\times\IRX{m}$ from $(1,1)$ to $(n,m)$.  For a vertex $(i,j) \in U$, let $\heightX{i,j}=\dY{p_i}{q_j}$ be its associated \emphw{weight}.  The path $\curveC$ is restricted to vertical and horizontal edges of the grid.  The \emphw{width} of $\curveC$ is
    \begin{math}
        \widthY{\curveC}{\curveA,\curveB} = \max_{(i,j) \in \VX{\curveC}} \heightX{i,j}.
    \end{math}
    The \emphi{discrete \Frechet distance} between $\curveA$ and $\curveB$ is the minimum of the width of any morphing between the two sequences.
\end{definition}

\subsubsection{The retractable discrete \Frechet distance}

\begin{definition}
    Let $\curveC = z_1, \ldots, z_{n+m-1} \in U$ be a discrete morphing between two sequences $\curveA = p_1, \ldots, p_n$ and $\curveB = q_1, \ldots, q_m$.  Let $D(i,j) = \max_{k: i < k < j} \heightX{z_k}$ be the inner width of the morphing $\curveC(i,j) = z_{i}, z_{i+1}, \ldots, z_{j}$. For simplicity, assume all pairwise distances $\dY{p_i}{q_j}$ are distinct.  The unique point $\btl(i,j)$ realizing the minimum $D(i,j)$, overall possible grid monotone paths $\curveC'$ between $z_i$ and $z_j$, is the \emphw{bottleneck} between $i$ and $j$. The \emphi{retractable discrete \Frechet morphing} between $\curveA$ and $\curveB$ is the unique morphing $\curveC$, such that
    \begin{compactenumi}
        \item $z_1 = (1,1)$, $z_{n+m-1} = (n,m)$, and
        \item for all $i < j$, we have $\btl( i,j ) \in \curveC(i,j)$.
    \end{compactenumi}
\end{definition}

Although the weights here are defined on the vertices, it is easy enough to interpret them as being on the edges (by, for example, assigning a grid edge $z \rightarrow z'$ the weight $\max( \heightX{z}, \heightX{z'} )$.  Plugging this (implicit) graph into \lemref{retractable} readily yields the following result.

\begin{lemma}
    \lemlab{f_d_r_algorithm}%
    Given two sequences $\curveA = p_1, \ldots, p_n$ and $\curveB = q_1, \ldots, q_m$, the retractable discrete \Frechet morphing between $\curveA$ and $\curveB$ can be computed in $O( n m \log (nm ))$ time.

    More generally, if the \Frechet distance between $\curveA$ and $\curveB$ is $\ell$, and
    \begin{equation*}
        \tau = \cardin{ \Set{(p,q)}{ p \in \curveA, q\in \curveB,
              \dY{p}{q} \leq \ell} },
    \end{equation*}
    then the running time of the algorithm is $O( \tau \log \tau)$.
\end{lemma}

\begin{remark}
    \begin{compactenumA}
        \item The first running time bound of \lemref{f_d_r_algorithm} is a worst-case bound, and for many realistic inputs, it is much smaller if the (discrete) \Frechet distance between them is small compared to the diameter of the two curves.  See \figref{near_linear} for an example.

        \smallskip%
        \item Note that if one has to explore a large fraction of the grid (i.e., $\tau = \Omega( n m)$), then the retractable \Frechet algorithm is slower than the standard discrete \Frechet algorithm, because of the use of a heap.
    \end{compactenumA}
\end{remark}

\section{The \VEFrechet distance}

In this section, we give a formal definition of the \VEFrechet distance, show some of its basic properties in relation to other variants of the \Frechet distance, and briefly discuss practical considerations.

\subsection{Definition and basic algorithm}

The input is two polygonal curves $\curveA = p_1 p_2 \cdots p_n$ and $\curveB = q_1 q_2 \cdots q_m$. This induces the free space diagram $\doggy =\doggyY{\curveA}{\curveB}$ (\defref{f_s_d}), which is a rectangle $R$, partitioned into a non-uniform $(n-1) \times (m-1)$ grid, where the $i$\th edge $p_ip_{i+1} \in \curveA$ and the $j$\th edge $q_j q_{j+1} \in \curveB$, induces the $(i,j)$ grid cell $\CellY{i}{j}$.

\begin{definition}
    For two points $u, v \in \Re^d$, let $\dirY{u}{v} = \frac{ v - u}{\dY{v}{y}}$ be the (unit length) \emphi{direction} vector from $u$ to $v$.
\end{definition}

Let $\vecP_i = \dirY{p_i}{p_{i+1}}$ and $\vecQ_j = \dirY{q_j}{q_{j+1}}$.  Let $x_i = \lenX{ \curveA[p_1,p_i] }$ and $y_j = \lenX{ \curveB[q_1,q_{j}]}$.  The grid cell $\Cell_{i,j}$ is the subrectangle $[x_i, x_i + \dY{p_i}{p_{i+1}}] \times [ y_j, y_j + \dY{q_j}{q_{j+1}}]$. The elevation function for any point $(x,y)$ inside $\CellY{i}{j}$ is given by
\begin{equation*}
    \elevY{x}{y}
    =
    \dY{\pi(x)}{-\sigma(y)}
    =
    \dY{p_i + (x - x_i) * \vecP_i
    }{\,-\, q_j - (y
       - y_j) * \vecQ_j\bigr.}.
\end{equation*}

The function $\elevY{x}{y}$ is a smooth convex function.%
\footnote{Specifically, it is a square root of a paraboloid, as a tedious but straightforward calculation shows.} More pertinent for our purposes, is that the minimum of the elevation function along each boundary edge of $C_{i,j}$ is unique, and corresponds to the distance of a vertex of one curve (say $p_i \in \curveA$) to the closest point on an edge (i.e., $q_j q_{j+1}$) of the other curve. This minimum along an edge of $\CellY{i}{j}$ is a \emphi{portal}. The portals on the left and bottom edges of $\CellY{i}{j}$ are the \emphw{incoming portals}, and the other two are \emphw{outgoing portals}.

The standard \Frechet distance asks for computing the $x/y$-monotone path between the opposite corners of the free space, while minimizing the maximum height along this path. The corresponding \Frechet morphing has the property that its intersection with a cell is either a segment (or empty). We restrict the \Frechet morphing further, by requiring that inside a cell, the morphing must be a straight segment connecting two portals of the cell. As a result, we have to give up on the monotonicity requirement (more on that shortly).

Importantly, computing the new \VEFrechet distance is now a graph search problem, in the following graph.

\begin{figure}[h]
    \phantom{}%
    \hfill%
    \includegraphics[width=0.2\linewidth]{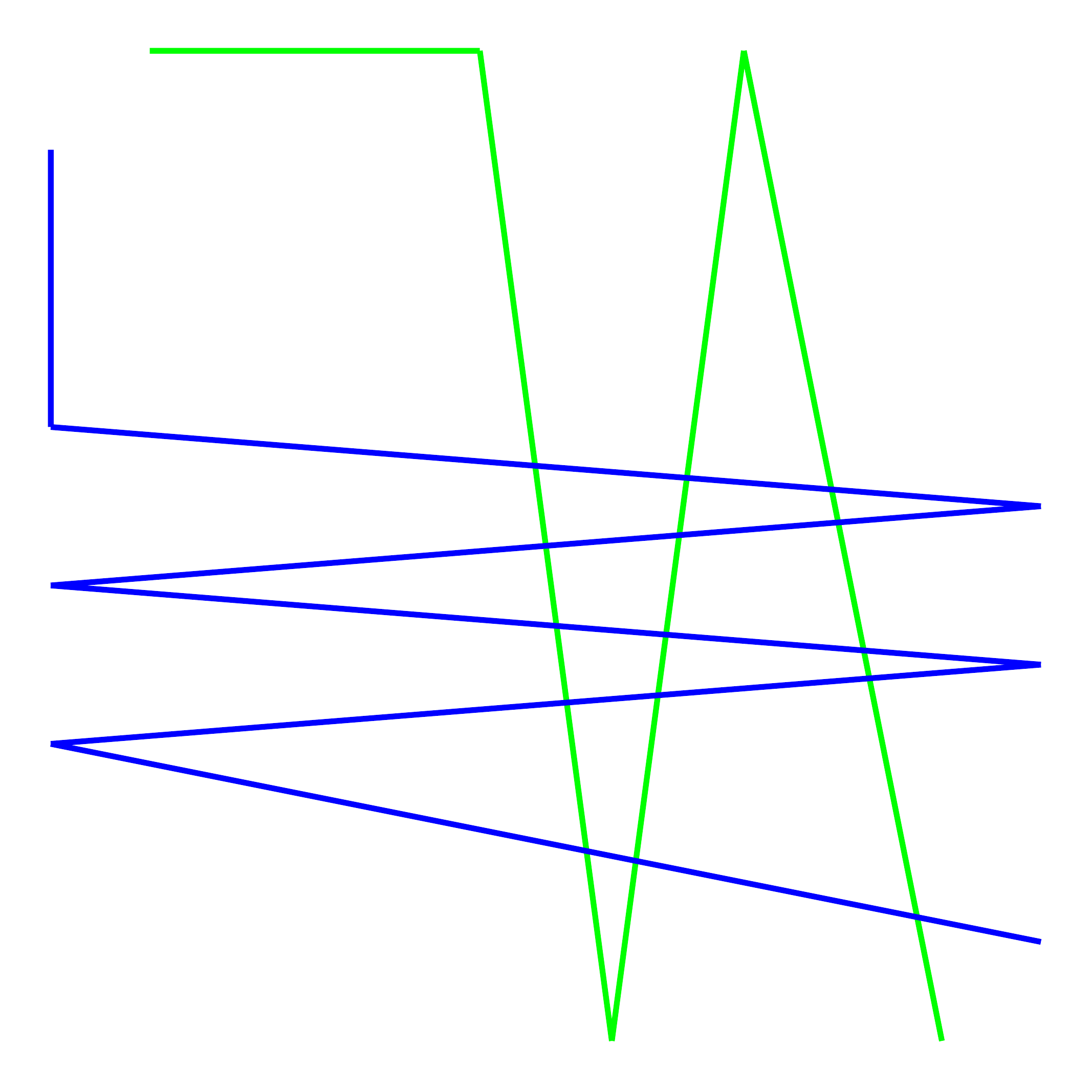}%
    \hfill%
    \includegraphics[width=0.38\linewidth]{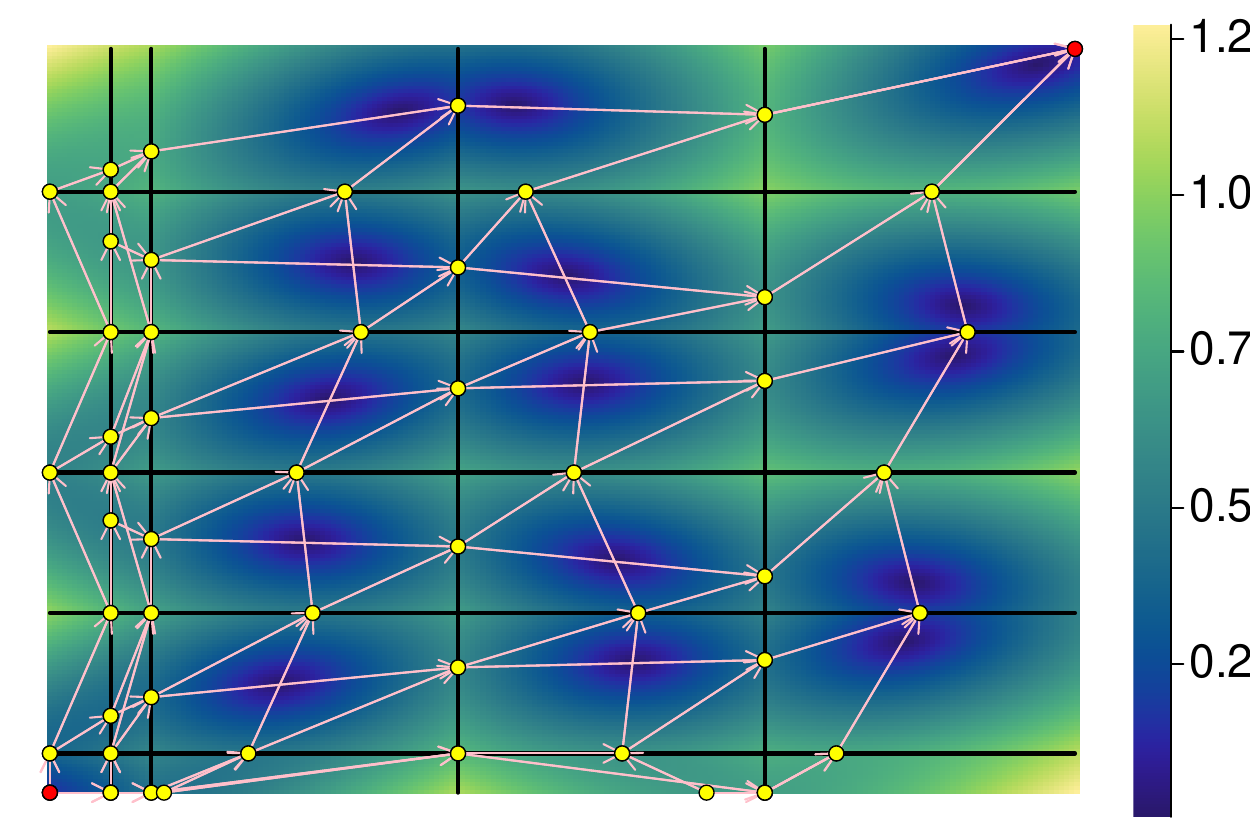}%
    \hfill%
    \phantom{}%
    \caption{Left: Two curves. Right: Their elevation function, and the associated graph.  See \href{\baseUrlX{10}}{here} for more info.  }
    \figlab{v_e_g_2}
\end{figure}

\begin{definition}
    \deflab{v_e_graph}%
    Given two curves $\curveA = p_1 p_2 \cdots p_n$ and $\curveB = q_1 q_2 \cdots q_m$, their \emphi{\VE graph} $\GVEC = \GVEY{\curveA}{\curveB}$ is a \DAG defined over the grid of the free space diagram $\doggy = \doggyY{\curveA}{\curveB}$, where every grid cell has four vertices on its boundary (some of them might coincide), corresponding to the portals on the edges.  Here, every cell has edges from each of its incoming portals to each of its outgoing portals (i.e., four edges in total when the portals do not coincide).  In addition, there are two special vertices -- the start vertex $s$ (i.e., bottom left corner of $\doggy$), and target vertex $t$ (i.e., top right corner of $\doggy$).  The portals on the grid edges adjacent to $s$ and $t$ are moved to $s$ and $t$, respectively.  Every vertex $v$ of this graph, has an associated location $(x,y) \in \doggy$, which in turn corresponds to two points $p_v = \curveA(x)$ and $q_v = \curveB(y)$. The \emphi{height} of $v$ is the elevation of $(x,y)$ -- that is $\dY{p_v}{q_v}$.
\end{definition}

For an example of the $\GVEC$ graph, see \figref{v_e_g_1}.  The \VEFrechet distance between $\curveA$ and $\curveB$ is simply the bottleneck distance between $s$ and $t$ in $\GVEC$. We might as well compute the \defrefY{retractable}{retract} version.

\begin{definition}
    Given two polygonal curves $\cA$ and $\cB$, the \emphi{retractable \VEFrechet morphing} is the unique morphing $\mrp$ realized by the retractable path in $\GVEY{\curveA}{\curveB}$ from $s$ to $t$, see \defref{retract}. The \emphi{\VEFrechet distance} is the maximum elevation of any point along $\mrp$, denoted by $\distVEFr{\cA}{\cB}$.
\end{definition}

Putting the above together, and using the algorithm of \lemref{retractable}, we get the following.
\begin{theorem}
    \thmlab{v_e_main}%
    Given two polygonal curves $\curveA$ and $\curveB$ with $n$ and $m$ vertices respectively, their \VEFrechet distance, and the corresponding retractable \VE \Frechet morphing $\morphVE{\cA}{\cB}$ that realizes it, can be computed in $O(n m \log (nm))$ time.
\end{theorem}

We emphasize that the running time bound in this theorem is pessimistic, and in many cases, the algorithm is significantly faster.  For an example of the output of the algorithm, see \figref{v_e_g_3}.

\begin{figure}[h]
    \centerline{\includegraphics[width=0.6\linewidth]{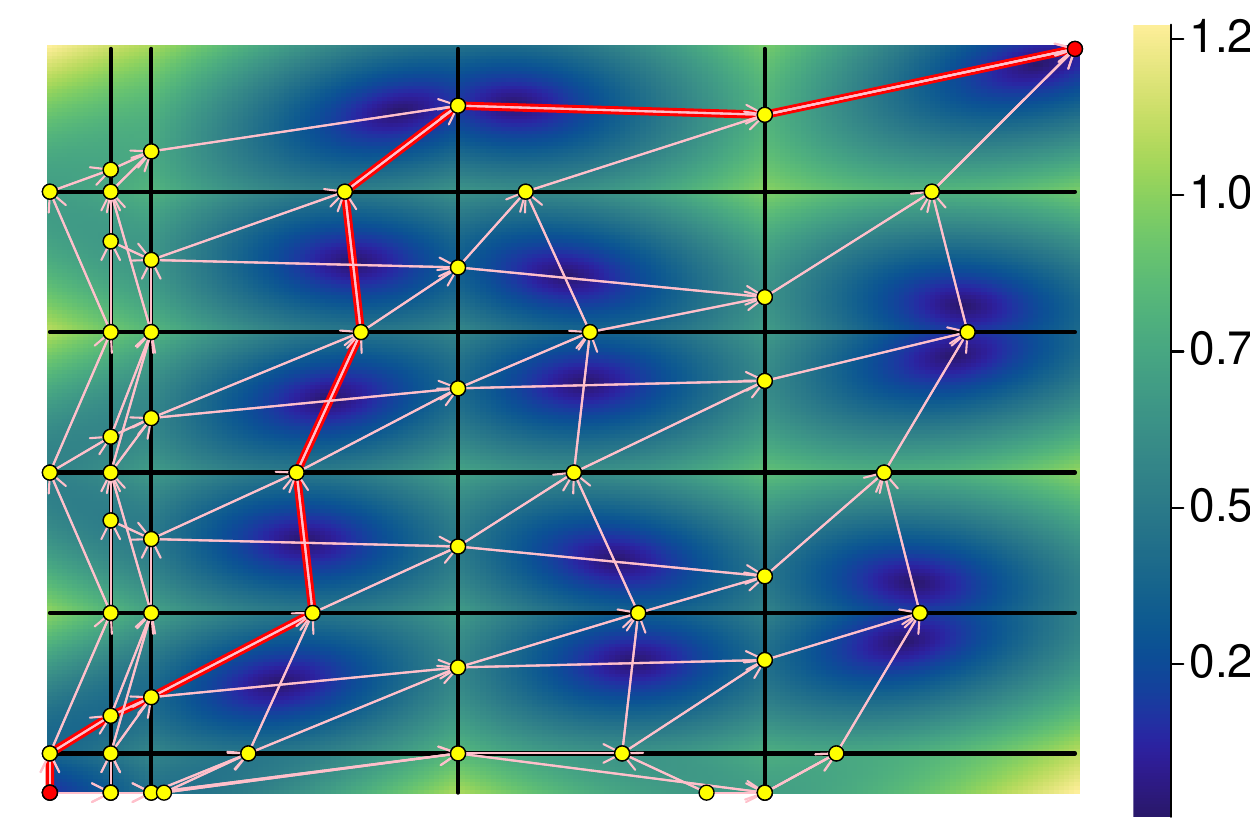}}
    \caption{The \VEFrechet morphing for the two curves from \figref{v_e_g_1}. Note that the solution is ``slightly'' not $x/y$-monotone. For animation of this morphing, follow the \href{https://frechet.xyz/examples/10/}{link}.  }
    \figlab{v_e_g_3}
\end{figure}

\subsubsection{Reachable complexity}

\begin{definition}
    For two curves $\cA$ and $\cB$, consider the \VE graph $\GVEC = \GVEY{\cA}{\cB}$. For a threshold $r \geq 0$, let $\relCZ{\leq r}{\cA}{\cB}$ be the number of vertices of $\GVEC$ that are reachable from the start vertex via paths with maximum elevation at most $r$. The quantity $\relCZ{\leq r}{\cA}{\cB}$ is the \emphi{$r$-reachable complexity} of $\cA$ and $\cB$.
\end{definition}

A similar notion (relative free space complexity) was used by Driemel \etal \cite{dhw-afdrc-12}. If $r = \distVEFr{\cA}{\cB}$, then the algorithm for computing the \VE morphing only explores vertices in the $r$-reachable region, which readily implies the following.

\begin{corollary}
    \corlab{v_e}%
    Given two polygonal curves $\curveA$ and $\curveB$ with $n$ and $m$ vertices respectively, their \VEFrechet distance, and the corresponding retractable \VE \Frechet morphing $\morphVE{\cA}{\cB}$ that realizes it, can be computed in $O( n + m + \relCC \log \relCC )$ time, where $\relCC = \relCZ{\leq r}{\cA}{\cB}$ and $r= \distVEFr{\cA}{\cB}$.
\end{corollary}

\begin{remark}
    The above implies that the worst case, for the above algorithm, is when $N = n m$. This happens, for example, (somewhat counterintuitively) when $\cA$ and $\cB$ are far from each other, and their \VEFrechet distance is realized by the distance of their endpoints (say, at the start of the morphing). This somewhat absurd situation can happen for real-world inputs -- follow this \href{\baseUrlX{11}}{link} for an example.  Furthermore, this quadratic behavior is probably unavoidable, as computing the \VEFrechet distance probably requires quadratic time in the worst case. We offer some approaches to avoid this for real-world input later on.

    Fortunately, for many inputs where the two curves are close to each other (e.g., their \VEFrechet distance is significantly smaller than their diameters), the running time of the algorithm of \corref{v_e} is $O\bigl( (n+m) \log (n+m)\bigr)$, see \figref{near_linear} for an example.
\end{remark}

\subsection{Basic properties of the \VE \Frechet morphing/distance}

\begin{lemma}
    \lemlab{lowerbound}%
    For any two curves $\cA$ and $\cB$, we have that $\distFr{\cA}{\cB} \geq \distVEFr{\cA}{\cB}$.
\end{lemma}

\begin{proof}
    Consider the \Frechet morphing $\mrp$ realizing the \Frechet distance $\alpha = \distFr{\cA}{\cB}$. Track the path $\mrp$ in the free space diagram, and inside each cell snap it to the portals of the corresponding edges it crosses. Clearly, this results in a valid morphing that is a valid path in $\GVEC$ from the source to the target, and its bottleneck value is no bigger than the original value (as by definition portals were the minimum elevation points along the corresponding edges). However, $\beta =\distVEFr{\cA}{\cB}$ being the minimum bottleneck path in this graph from the start vertex to the target vertex, is definitely not bigger than $\alpha$.
\end{proof}

\begin{observation}\obslab{mono}
    The \VEFrechet morphing is monotone over the vertices.

    Specifically, let $\mrp$ be the \VEFrechet morphing between $\curveA$ and $\curveB$. If we project $\mrp$ on the $x$-axis, denoted by $x(\mrp)$, we get a path from $0$ to $\lenX{\curveA}$ on the real line. Note, that this path is not necessarily monotone, but importantly, for any vertex $v \in \curveA$ it is monotone. Specifically, let $\alpha$ be the distance of $v$ along $\curveA$ from its start. Then, the path $x(\mrp)$ crosses from one side of $\alpha$ to the other side, exactly once (assuming $v$ is an internal vertex).  To see why this is true, observe that any morphing must cross the vertical line $x = \alpha$. This line has vertices along it, and observe that in $\GVEY{\curveA}{\curveB}$ it is only possible to cross this line by using one of these vertices, as no edge crosses from one side of this line to the other (i.e.\ the vertices on this line are a separator).  However, vertices along this line have outgoing edges only to vertices that are to the right (or remain on the same $x$, but these target vertices are on grid edges adjacent to the right of this line).
\end{observation}

\begin{definition}
    Given a polygonal curve $\curveA=p_1 p_2 \cdots p_n$, a \emphi{refinement} of $\curveA$ is a polygonal curve $\curveA'=p_1 Q_1 p_2 Q_2 p_3 \cdots Q_{n-1} p_n$, where for all $1\leq i<n$, $Q_i$ is a (possibly empty) sequence of points occurring in order along the (directed) edge $p_ip_{i+1}$.
\end{definition}

\begin{observation}
    If $\cA', \cB'$ are refinements of $\curveA, \curveB$ respectively, then $\distFr{\cA'}{\cB'} = \distFr{\cA}{\cB}$ -- indeed, the additional vertices from a refinement do not change the underlying curve.

    However, $\distVEFr{\cA'}{\cB'} \geq \distVEFr{\cA}{\cB}$ -- indeed, refinement corresponds to adding vertical and horizontal lines in the free space diagram that the morphing can cross only once (i.e., the morphing is monotone at the vertices). As such, the morphing realizing $\distVEFr{\cA'}{\cB'} $ is determined by more constraints, and it is thus (potentially) larger.
\end{observation}

\section{Computing the regular \Frechet distance}

\subsection{Direct monotonization}

The basic idea to compute the \Frechet distance using the new algorithm is to refine the curves till their \VEFrechet morphing becomes monotone.

\begin{observation}
    If the \VEFrechet morphing is $x/y$-monotone, then it realizes the (regular) \Frechet distance between the curves. Indeed, if the morphing is monotone, then it is a valid morphing considered by \defref{width_f} for the \Frechet distance, and thus $\distVEFr{\cA}{\cB} \geq \distFr{\cA}{\cB}\geq \distVEFr{\cA}{\cB}$, where the latter inequality is from \lemref{lowerbound}.
\end{observation}

\subsubsection{How morphing encodes the motion.}
A \VEFrechet morphing $\mrp$ in a free space diagram $\doggy = \doggyY{\cA}{\cB}$ of two curves $\cA$ and $\cB$ with $n$ and $m$ vertices, respectively, is a polygonal path with $t = n+ m +O(1)$ vertices. Thus, a morphing can be represented as $\mrp = (x_1,y_1), \ldots, (x_t,y_t)$. A consecutive pair of vertices of the morphing $(x_i, y_i), (x_{i+1}, y_{i+1})$ corresponds to two directed subsegments $\seg \subseteq \cA$ and $\segA \subseteq \cB$. The matching of these two subsegments encodes a linear motion starting at the start points and ending at the endpoints.  See \figref{morphing}.

\begin{figure}
    \phantom{}%
    \hfill%
    \includegraphics[width=0.25\linewidth]{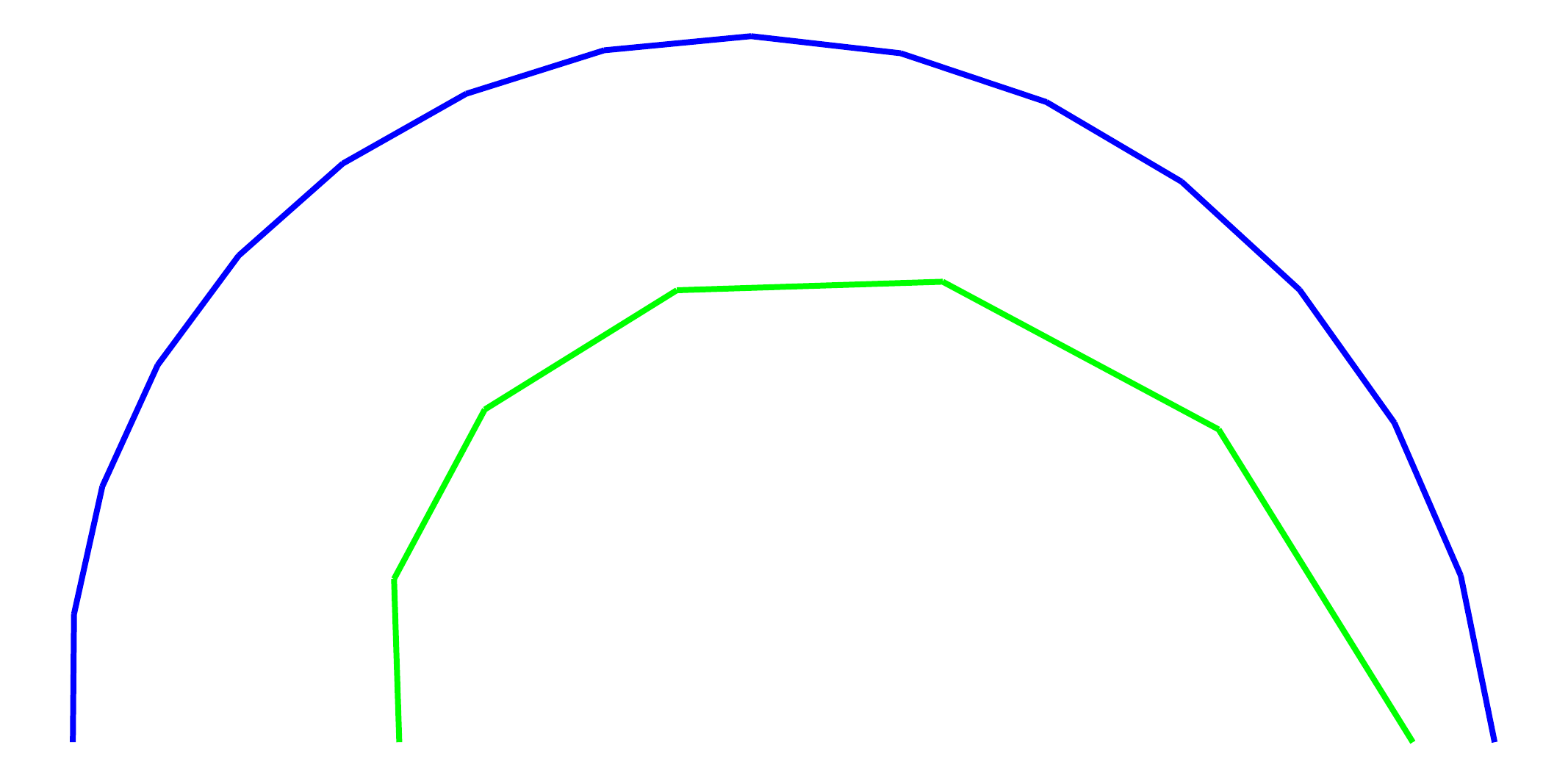} \hfill%
    \includegraphics[width=0.25\linewidth]{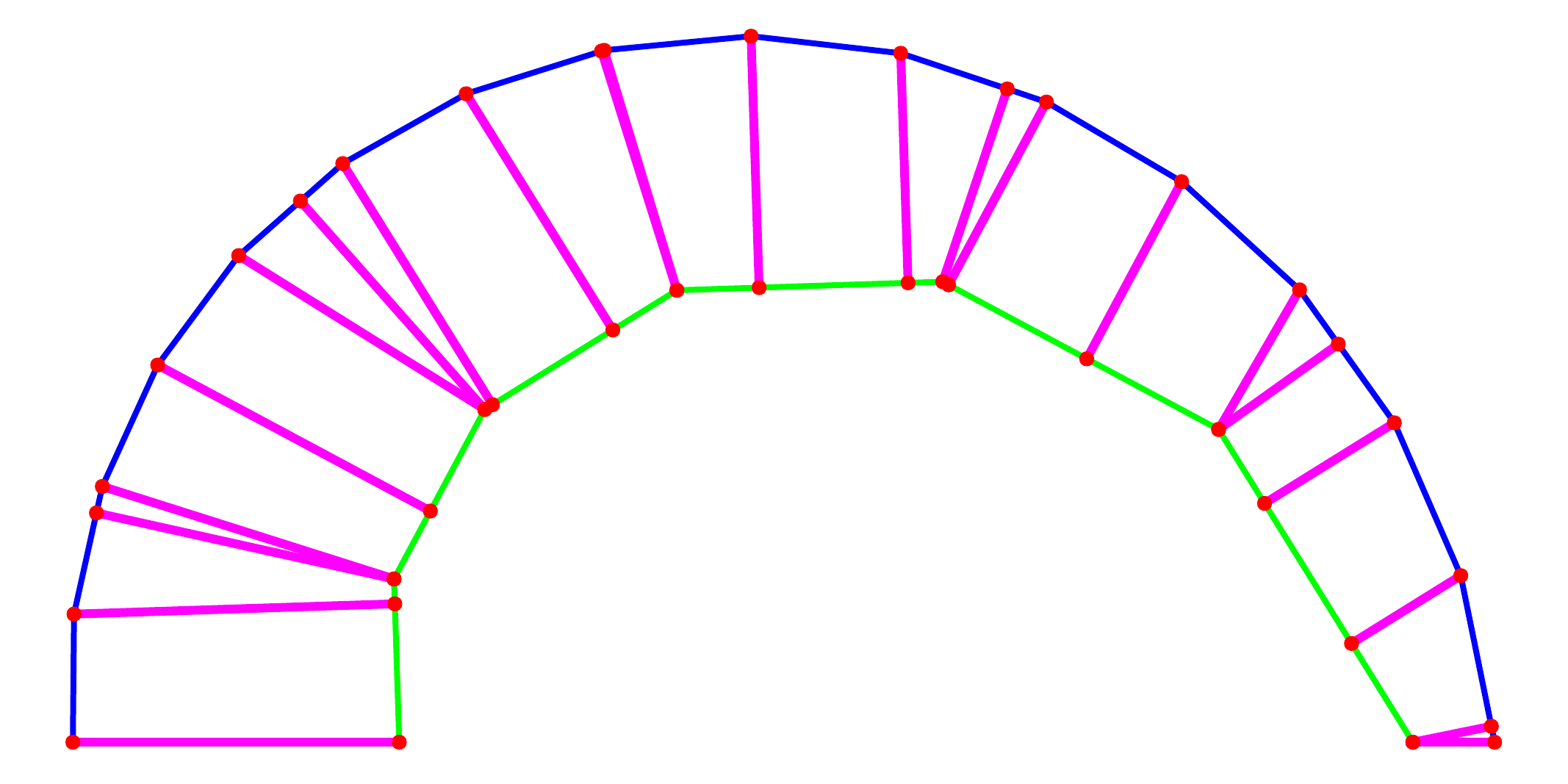} \hfill%
    \phantom{}
    \caption{Left: Two curves. Right: Their \Frechet morphing. Since the \Frechet morphing computed inside each cell of the free space diagram is a segment, it is enough to mark all the critical configurations when the morphing enters/leaves a cell. Each critical configuration is depicted on the right by a segment connecting the two points being matched. In between two such configurations, the morphing is the linear interpolation between the two configurations.  See \href{\baseUrlX{15/}}{here} for an animation of this specific morphing.}%
    \figlab{morphing}
\end{figure}

\begin{definition}
    Given a morphing $\mrp = (x_1,y_1), \ldots, (x_t,y_t)$, let $X_i = \max_{j=1}^i x_j$ and $Y_i = \max_{j=1}^i y_j$, for $i=1,\ldots, t$. Let $\monoX{\mrp} = (X_1,Y_1), (X_2,Y_2), \ldots, (X_t, Y_t)$ denote the \emphi{direct monotonization} of $\mrp$ (which can be computed in $O(t)$ time).
\end{definition}

\begin{observation}%
    \obslab{interval}%
    Let $\VEmorphing$ be the \VEFrechet morphing between $\cA$ and $\cB$.  Then $\distFr{\cA}{\cB} \in [\distVEFr{\cA}{\cB}, \WidthX{\monoX{\VEmorphing}}]$, where the lower bound follows from \lemref{lowerbound} and the upper bound follows as $\monoX{\VEmorphing}$ is a monotone morphing and thus is a valid morphing considered by \defref{width_f} for the \Frechet distance.
\end{observation}

For many natural inputs $\WidthX{\VEmorphing}=\WidthX{\monoX{\VEmorphing}}$. Specifically, this occurs if either $\VEmorphing$ is monotone to begin with or if the ``zigzagging'' in $\VEmorphing$ is happening in parts where the leash is relatively short, and thus the fixup due to direct monotonization does not affect the global bottleneck. Note that in this case, $\distVEFr{\cA}{\cB}=\WidthX{\monoX{\VEmorphing}}$, and so by \obsref{interval}, $\monoX{\VEmorphing}$ is an optimal morphing realizing $\distFr{\cA}{\cB}$.

Given a morphing $\mrp$, one can compute its width in linear time, as by convexity it is realized by the elevation of one of the vertices of $\mrp$.
Thus, a first step to trying to compute $\distFr{\cA}{\cB}$ is to compute the \VEFrechet morphing $\VEmorphing$ between $\cA$ and $\cB$, and then directly monotonize it. If $\WidthX{\VEmorphing} = \WidthX{\monoX{\VEmorphing}}$, then we have computed the \Frechet distance, and so we are done. Otherwise, we will use refinement to narrow the interval $[\distVEFr{\cA}{\cB}, \WidthX{\monoX{\VEmorphing}}]$, as described below.

\subsubsection{Monotonization via refinement}
\seclab{monotone_refinement}

\myparagraph{The lower/upper bound.}

Let $\VEmorphing$ be the \VEFrechet morphing between $\cA$ and $\cB$.
By \obsref{interval}, $\distFr{\cA}{\cB}\in [\distVEFr{\cA}{\cB}, \WidthX{\monoX{\VEmorphing}}]$.
Thus, we have initial upper and lower bounds for $\distFr{\cA}{\cB}$. The idea is to improve, repeatedly, these bounds until they are equal. At this point, the computed morphing realizes the \Frechet distance between the two curves.

\begin{observation}
    Let $\eps \in (0,1)$. One can refine $\cA$ and $\cB$ into curves $\cA'$ and $\cB'$, with \VEFrechet morphing $\VEmorphing'$, such that
    \[
        \WidthX{\monoX{\VEmorphing'}}-\distVEFr{\cA'}{\cB'} \leq \eps
    \]
    Indeed, create $\cA'$ and $\cB'$ by introducing vertices along $\cA$ and $\cB$ respectively, such that no edge has length exceeding $\eps/2$.  Recall that by \obsref{mono}, $\VEmorphing'$ is already monotone over vertices, and thus $\monoX{\VEmorphing'}$ does not change the edges from the curves that points in the morphing are mapped to.  Thus, the direct monotonizations of the \VE morphings of $\cA'$ and $\cB'$ can increase the distance by at most twice the maximum edge length.
\end{observation}

Thus, in the limit, one can compute the \Frechet distance from a refinement of the two curves. Fortunately, there is a finite refinement that realizes the \Frechet distance.
\begin{lemma}
    \lemlab{refinement_exact}%
    Let $\cA$ and $\cB$ be two polygonal curves of total complexity $n$. Then, there is a set $S$ of $O(n^2)$ vertices, such that if we refine $\cA$ and $\cB$ using $S$, then for the resulting two curves, $\cA'$ and $\cB'$, we have that
    \begin{math}
        \distVEFr{\cA'}{\cB'} \leq \distFr{\cA}{\cB} = \WidthX{ \monoX{\morphVE{\cA'}{\cB'}} }.
    \end{math}
\end{lemma}
\begin{proof}
    Let $r=\distFr{\cA}{\cB}$, and introduce a vertex in the interior of an edge of $\cA$ if this interior point is in distance exactly $r$ from some vertex of $\cB$, and vice versa. Let $\cA'$ and $\cB'$ be the resulting refined curves. Consider the grid $H$ of $\doggyY{\cA'}{\cB'}$, and any edge $e$ of this grid. The value of the elevation function on $e$ is an interval: $\elevX{e} = \Set{\elevX{p}}{ p \in e}$, and this interval cannot contain $r$ in its interior (if it did, we would have broken it by introducing a vertex at the corresponding interior point).
    Thus, any edge in $H$ whose elevation interval has a value above $r$ cannot contain any interior point with value below $r$. Thus, such an edge is infeasible (in its interior) for the \VEFrechet morphing $\morphVE{\cA'}{\cB'}$, as we know $\distVEFr{\cA'}{\cB'} \leq r$. Namely, every edge of the free space is either entirely feasible (for $r$) or infeasible.

    \begin{figure}[h]
        \centerline{\includegraphics{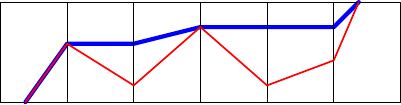}}
        \caption{The red path has a zigzag in this row. The blue path depicts the monotonized path.}
        \figlab{monotone}
    \end{figure}

    Now, if $\mrp = \morphVE{\cA'}{\cB'}$ is monotone then we are done.  Otherwise, the portions of $\mrp$ that are not monotone can be broken into a sequence of ``zigzags'' in columns and rows. Consider such a (say) row with non-monotonicity. The path (i.e., morphing) enters the row from the bottom on the left side, and leaves on the right from the top, see \figref{monotone}. The monotonized portion of this path inside this row enters and leaves through the same extreme portals, but importantly, it uses the same edges. All these edges have maximum elevation at most $r$, and thus the width of the monotonized morphing is at most $r$ (and thus, equal to $r$), establishing the claim.
\end{proof}

\begin{remark}
    \remlab{b_refinement}%
    It is straightforward to check if a given morphing $\mrp$ is monotone. If it is not, then compute for each curve the portions where this morphing travels on them backward, and introduce a vertex in the middle of each backward portion. Now, we recompute a morphing on the refined curve. Clearly, this reduces the ``back and forth'' non-monotone motion on the two curves. One can naturally repeat this process several times till the non-monotonicity disappears or becomes so small that it can be ignored.  We refer to this as \emphi{bisection refinement}.
\end{remark}

\subsection{Computing the regular \Frechet distance}

The above suggests natural strategies for computing the \Frechet distance:
\begin{compactenumI}
    \smallskip%
    \item First, compute the \VEFrechet morphing.  We now apply the bisection refinement, described in \remref{b_refinement}, to the resulting morphing (recomputing the morphing using \VEFrechet if needed).  In practice, a few rounds of this bisection refinement make these portions of the morphing involve leashes shorter than the bottleneck (even after direct monotonization), and we end up with the \Frechet distance.
\end{compactenumI}
\medskip%

\begin{remark}
    Underlying the bisection refinement strategy is the implicit assumption that the matching induced by the \VEFrechet morphing is (relatively) \emphi{stable}, and matches roughly the same portions of the curves after refinement. This seems to be the case for all the curves we considered. Understanding this stability might be an interesting venue for future research.
\end{remark}

\begin{compactenumI}[resume]
    \smallskip%
    \item Another (more systematic) approach to refinement is to consider all the vertices $V$ (say) of $\cA$ involved in a subpath of the morphing that is not monotone on an edge $e$ of $\cB$. We refine $e$ along all the points that are intersections of bisectors of pairs of points of $V$ with $e$.  Since repeated application of this process converges to the refinement set of \lemref{refinement_exact} (the set used in this lemma is too expensive to compute explicitly), this implies that sooner or later this yields the \Frechet distance (after direct monotonization).  Essentially, this performs a search of the monotonicity events, see \secref{event_types}, that are relevant to the current morphing.
\end{compactenumI}

\subsection{Speed-up via simplification}

A natural approach to try and avoid the quadratic complexity of computing the \VEFrechet distance is to use simplification of the input before computing the distance. Here, we explore how to use simplification so that we can compute the exact distance.  We start with a simple linear-time simplification algorithm \cite{ahkww-fdcr-06,dhw-afdrc-12}.

\begin{algorithm}
    \alglab{simplification}%
    For a polygonal curve $\cA= p_1, p_2, \ldots, p_n$, and a parameter $\delta>0$, the algorithm marks the initial vertex $p_1$, and sets it as the current vertex. The algorithm repeatedly scans $\cA$ from the current vertex till reaching the first vertex $p_i$ that is at least $\delta$ distance from the current vertex. It marks $p_i$, and sets it as the current vertex. The algorithm repeats this until reaching the final vertex, which is also marked. The resulting subcurve connecting the marked vertices $\cAs = \simpX{\curveA,\delta}$ is the \emphi{$\delta$-simplification} of $\cA$, and it is not hard to verify (and it is also well known \cite{dhw-afdrc-12}) that $\distFr{\cA}{\cAs}\leq \delta$.
\end{algorithm}

Observe that
\begin{math}
    \distFr{\cA}{\cB} - 2\delta%
    \leq%
    \distFr{\cAs}{\cBs}%
    \leq%
    \distFr{\cA}{\cB} + 2\delta.
\end{math}
(same holds for the \VE version), where $\cAs = \simpX{\cA,\delta}$ and $\cBs = \simpX{\cB,\delta}$.

\begin{lemma}%
    \lemlab{greedy_morph}%
    Given a curve $\cA$ with $n$ vertices, and a simplified curve $\cAs$ computed by applying \algref{simplification} to $\cA$, with a parameter $\delta$, one can compute, in $O(n)$ time, a \Frechet morphing between $\cA$ and $\cAs$ of width $\ell \leq 2\delta$, such that $\distFr{\cA}{\cAs} \leq \ell$.
\end{lemma}
\begin{proof}
    Consider the case that $\cAs$ is a single segment $s$.  Consider first the ``silly'' morphing $\mrp'$, that matches every vertex of $\cA$ to its nearest neighbor on $s$ (in between, it interpolates linearly). Clearly, this is a weak morphing between $\cA$ and $\cAs$, but it is also clearly optimal. Consider the monotone version $\mrp$ of $\mrp'$ that one gets from $\mrp'$ by never moving back along $s$, and let $\ell$ be its width.  It is not hard to verify that $\ell \leq 2\delta$.

    Now remove the assumption that $\cAs$ is a single segment.  By construction, $\cAs$ is a subsequence of $\cA$, including the start and end vertices. Thus, the edges of $\cAs$ partition $\cA$ into pieces, which start and end at the corresponding endpoints of the edge. Therefore, in this more general case, we can apply the above argument to each edge of $\cAs$ and its corresponding piece of $\cA$, and then concatenate all these morphings together.
\end{proof}

Given a $\delta$-simplification $\cAs$, the above gives us a fast way to compute a morphing between $\cA$ and $\cAs$. This morphing may not be optimal, but its error is bounded by $2\delta$, and we are free to set $\delta$ to our desired value (depending on which the simplified curve may have significantly fewer vertices).  This basic approach works for any subcurve specified by a subsequence of the vertices of $\cA$.

\begin{remark}
    A natural conjecture is that for any sub-curve $\cAs$ of $\cA$ defined by a subset of the vertices of $\cA$ (including the same start and end vertices), the above quantity $\distFr{\cA}{\cAs} \leq \ell \leq 2 \distFr{\cA}{\cAs}$. To see that this is false, let $m > 1$ be an integer, and let $\delta = 1/2m$, and consider the one dimensional curve and its subcurve:
    \begin{align*}
      \cA &= 0, 1, \delta, 1-\delta, 2\delta, 1-2\delta, \ldots,
            1/2-\delta, 1/2+\delta, 1/2.\\
      \cAs &= 0, \phantom{1, \delta,} 1-\delta, 2\delta, 1-2\delta, \ldots,
             1/2-\delta, 1/2+\delta, 1/2.
    \end{align*}
    The above algorithm would compute $\ell \approx 1$ (or $\ell \approx 1/2$ if it computes the exact \Frechet distance between every segment of the simplification and its corresponding subcurve), but $\distFr{\cA}{\cAs} = O( \delta)$.
\end{remark}

\subsubsection{Combining morphings}

\myparagraph{Getting a monotone morphing via simplification.}  The input is the two curves $\cA$ and $\cB$, and a parameter $\delta$.  We compute $\cAs = \simpX{\cA,\delta}$ and $\cBs = \simpX{\cB,\delta}$. Next, compute
\begin{equation*}
    \mrp_1 = \monoX{ \morphVE{\cA}{\cAs} },
    \mrp_2 = \monoX{ \morphVE{\cAs}{\cBs} },
    \quad\text{and}\quad
    \mrp_3 = \monoX{ \morphVE{\cBs}{\cB} }.
\end{equation*}

For the sake of simplicity of exposition, think about these morphings $\mrp_i$ as piecewise linear $1$-to-$1$ functions. For example, $\mrp_1: [0,\lenX{\cA}] \rightarrow [0, \lenX{\cAs}]$. Thus, combining the above three morphings $\mrp = \mrp_3 \circ \mrp_2 \circ \mrp_1: [0,\lenX{\cA}] \rightarrow [0,\lenX{\cB}]$ yields the desired monotone morphing.

\begin{remark}
    In practice, one does not need to compute $\mrp_1 = \monoX{ \morphVE{\cA}{\cAs} }$, but one can do a ``cheapskate'' greedy morphing (which is monotone) as done by the algorithm of \lemref{greedy_morph}.
\end{remark}

The algorithm for combining two morphings is a linear-time algorithm somewhat similar in nature to the merge subroutine used in merge-sort. We omit the straightforward but tedious details.
\begin{lemma}
    Given two compatible $x/y$-monotone morphings $\mrp_1$ and $\mrp_2$ (i.e., polygonal curves in $\Re^2$), their combined morphing $\mrp_2 \circ \mrp_1$ can be computed in linear time.
\end{lemma}

\begin{remark}
    The above provides us with a ``fast'' algorithm for computing a morphing between two input curves $\cA$ and $\cB$ in time that is effectively near linear (if the simplification is sufficiently aggressive). Specifically, we simplify the two curves, and then compute the monotone \Frechet distance between the simplified curves by computing the \VEFrechet morphing and refining it till it becomes monotone. We then stitch the three morphings together to get a morphing for the original two curves. Clearly, by making the simplification sufficiently small, this converges to the optimal \Frechet morphing.

    The question is thus the following: Given a \Frechet morphing, can we test ``quickly'' whether it realizes the \Frechet distance between the two curves (without computing the \Frechet distance directly).
\end{remark}

\subsection{Computing the exact \Frechet distance using simplification}

\subsubsection{Morphing sensitive simplification}

The challenge in using simplification for computing the \Frechet distance is to be able to argue that the computed morphing is optimal. To this end, we are interested in computing a lower bound on the \Frechet distance from simplified curves, such that the upper/lower bounds computed match, thus implying that the computed solution is indeed optimal.

\myparagraph{A weak lower bound.}

Consider two curves $\cA$ and $\cB$, and their respective simplifications $\cAs$ and $\cBs$. By the triangle inequality, we have
\begin{equation*}
    \distFr{\cA}{\cB}
    \geq
    \lbY{\cA}{\cB} =
    \distFr{\cAs}{\cBs}
    -
    \distFr{\cA}{\cAs}
    -
    \distFr{\cB}{\cBs}.
\end{equation*}
The \emphi{lower bound} $\lb$ is usually easy to compute (or bound from below), since the simplifications usually provide an immediate upper bound on $\distFr{\cA}{\cAs}$ and $\distFr{\cB}{\cBs}$.

\begin{definition}
    For any $\psi >1$, a \Frechet morphing $\mrp$ of $\cA$ and $\cB$ is a \emphi{$\psi$-approximate morphing} if
    \begin{math}
        \distFr{\cA}{\cB} \leq \WidthX{\mrp} \leq \psi \distFr{\cA}{\cB}.
    \end{math}
    Thus $\lb = \WidthX{\mrp} / \psi$ is a lower bound on the \Frechet distance between $\cA$ and $\cB$.
\end{definition}

\begin{definition}
    For a morphing $\mrp \in \MrpY{\cA}{\cB}$, a lower bound $\lb \leq \distFr{\cA}{\cB}$, and a point $p=(x,y) \in \mrp$, the \emphi{slack} at $p$ is $\slackX{p} = \max( \lb - \elevX{p}, 0)$.  A point $p \in \mrp$ is \emphi{tight} if $\slackX{p} = 0$.
\end{definition}

Consider an optimal (say retractable) \Frechet morphing $\mrp$ between $\cA$ and $\cB$. If the lower bound is the \Frechet distance, then we expect only the point in the morphing realizing the bottleneck to be tight.  One can map the slack from the morphing back to the original curves. Specifically, for $v \in \cA \cup \cB$, let $\max_\mrp(v)$ be the maximum length leash attached to $v$ during the morphing $\mrp$ (this can correspond to the maximum elevation along a vertical or horizontal segment in $\mrp$).

\begin{definition}
    For a vertex $v \in \cA \cup \cB$, the \emphi{slack} of $v$ is
    \begin{math}
        \slackX{v}%
        =%
        \max ( \lb - \max\nolimits_\mrp(v), 0 ).
    \end{math}
\end{definition}

Intuitively, the slack of a vertex is how much one can move it around without introducing too much error -- observe that for some vertices the slack is zero, implying they cannot be moved. So, consider a curve $\cA = p_1, \ldots, p_n$, and a ``simplified'' subcurve of it $\cAs$. For a vertex $p_i \in \cA$, its \emphi{representative} in $\cAs$, denoted by $\repX{p_i}$, is the point $p_j$, such that $p_j$ is a vertex of $\cAs$, $j$ is maximal, and $j \leq i$.

\begin{definition}
    For some real number $\tau \geq 1$, the curve $\cAs$ is a \emphw{$(\mrp,\tau)$-sensitive simplification} (or simply \emphi{$\mrp$-simplification}) of $\cA$ if, for all $i$, we have that $\dY{p_i}{-\repX{p_i}} \leq \slackX{p_i}/\tau$.
\end{definition}

\subsubsection{Computing a lower-bound on the \Frechet distance}

Consider an edge $e = p_i p_j$ of the simplified curve. The \emphi{width} of the edge, denoted by $\wFX{e} = \distFr{p_ip_j}{ p_i p_{i+1} \ldots p_{j}}$, and consider the hippodrome
\begin{equation*}
    \hippoX{e}
    =
    e \oplus \wFX{e}
    =%
    \Set{ p \in \Re^d}{ \dSY{p}{e} \leq \wFX{e}},
\end{equation*}
where $\dSY{p}{e}$ is the distance of $p$ from $e$.  More generally, for a point $p \in e \subseteq \cAs$, we denote by $\wFX{p} = \wFX{e}$.

Let $\mrp_\cA:[0,\lenX{\cA}] \rightarrow [0,\lenX{\cAs}$, be a (monotone) morphing of $\cA$ to $\cAs$ ($\mrp_\cB$ is defined similarly for $\cB$ and $\cBs$).  There is a natural simplification of the elevation function $\elevY{x}{y}$. Specifically, we define
\begin{align*}
  \elevSY{x}{y}
  &=%
    \dY{\cAs(\mrp_\cA(x)\bigr)}{ - \cBs\bigl(\mrp_\cB(y)\bigr) }
    - \wFX{\cAs(\mrp_\cA(x)\bigr)}
    - \wFX{\cBs(\mrp_\cB(y)\bigr)}
  \\&%
  \leq%
  \dY{\cAs(\mrp_\cA(x)\bigr)}{ - \cBs\bigl(\mrp_\cB(y)\bigr) }
  - \dY{\cA(x)}{-\cAs(\mrp_\cA(x)\bigr)}
  - \dY{\cB(y)}{-\cBs(\mrp_\cB(y)\bigr)}
  \leq %
  \elevY{x}{y}.
\end{align*}
Namely, computing the \Frechet distance using $\elevCS$ (instead of $\elevC$), would provide us with a lower bound on the \Frechet distance between the two curves. Of course, this \Frechet distance computation can be done directly on the corresponding ``elevation'' function between the two simplified curves $\cAs$ and $\cBs$:
\begin{equation*}
    \forall (x,y) \in [0,\lenX{\cAs}] \times [0,\lenX{\cBs}]
    \qquad
    \elevDY{x}{y}
    =%
    \dY{\cAs(x)}{ - \cBs( y ) }
    - \wFX{\cAs(x\bigr)}
    - \wFX{\cBs(y\bigr)}.
\end{equation*}
Naturally, one can replace $\wFX{\cAs(x\bigr)}$ and $\wFX{\cBs(y\bigr)}$ by larger quantities, and still get the desired lower bound. Let $\distDFr{\cAs}{\cBs}$ denote this lower-bound on the \Frechet distance.

\subsubsection{Computing the exact \Frechet distance}

The above suggests an algorithm for computing the exact \Frechet distance. Start with a low-quality approximate morphing $\mrp$ between $\cA$ and $\cB$ (this can be computed directly by simplifying the two curves). Use this to get a more sensitive approximation to the two curves, and compute the \Frechet distance between the two simplified curves. This yields a morphing between the two original curves (i.e., upper bound), and a matching lower bound on the \Frechet distance between the two curves. If the two are equal, then we are done.

The critical property of this algorithm is that it never computes directly the exact \Frechet distance between the two original curves, which might be too large to handle in a reasonable time. The result is summarized in the following lemma.

\begin{lemma}
    Let $\cA$ and $\cB$ be two curves, and let $\mrp$ be a $\psi$-approximate morphing between them, for some $\psi > 1$.  Let $\tau \geq 1$ be some constant, and consider $(\mrp,\tau)$-sensitive simplifications $\cAs,\cBs$, of $\cA,\cB$, respectively. Let $\mrp_2: \cAs \rightarrow \cBs$ be the optimal \Frechet morphing between $\cAs$ and $\cBs$, and let $\mrp_1: \cA \rightarrow \cAs$, and $\mrp_3: \cBs \rightarrow \cB$, be the natural morphings associated with the simplifications. This gives rise to a natural morphing $\mrp' = \mrp_3 \circ \mrp_2 \circ \mrp_1$ from $\cA$ to $\cB$.  If $\WidthX{\mrp'} = \distDFr{\cAs}{\cBs}$, then $\mrp'$ realizes the optimal \Frechet distance between $\cA$ and $\cB$.  where $\distDFr{\cAs}{\cBs}$ is the \Frechet distance between $\cAs$ and $\cBs$ under the modified elevation function $\elevD$.
\end{lemma}

It is easy to verify that for $\tau$ large enough, the two simplified curves are the original curves, and the above would compute the \Frechet distance. Thus, the resulting algorithm is iterative -- as long as the above algorithm fails, we double the value of $\tau$, and rerun the algorithm, till success.

\subsection{The sweep distance}
\seclab{sweep_distance}

\subsubsection{Definitions}
Another measure of distance between curves is the \CDTW (Continuous Dynamic Time Warping) distance. This distance has a neat description in our setting -- given a morphing $\mrp$ between $\cA$ and $\cB$, we define the two functions
\begin{equation*}
    f(x) = \min_{y: (x,y) \in \mrp} \elevY{x}{y}
    \qquad\text{and}\qquad
    g(y) = \min_{x: (x,y) \in \mrp} \elevY{x}{y}.
\end{equation*}
In words, $f(x)$ (resp. $g(y)$) is the minimum length leash attached to the point $\cA(x)$ (resp. $\cB(y)$) during the morphing of $\mrp$. We define the \emphi{warping cost} of $\mrp$ to be
\begin{equation}
    \costX{\mrp}
    =%
    \int_{x=0}^{\lenX{\cA}} f(x) \dX{x}
    +
    \int_{y=0}^{\lenX{\cB}} g(y) \dX{y}.
    \eqlab{integrals}
\end{equation}

The \emphi{Continuous Dynamic Time Warping} (\emphi{\CDTW}), between the two curves is
\begin{equation*}
    \dTWY{\cA}{\cB} = \min_{\mrp \in \MrpMY{\cA}{\cB}} \costX{\mrp}.
\end{equation*}

Computing the \CDTW between curves is not easy, and only partial results are known -- see the introduction for details. We point out, however, that our available machinery readily leads to computing an upper bound on the \CDTW. Furthermore, this upper bound can be made to converge to the \CDTW distance via simple refinements of the two curves.
In particular, for a well-behaved morphing, the integrals of \Eqref{integrals} can be computed exactly.

\subsubsection{The sweep distance}

\myparagraph{The cost of morphing along a segment inside a cell.} %
Given two directed segments $\seg = \pA_1 \pA_2$ and $\seg' =\pB_1\pB_2$, the natural linear morphing between them is the one where the two points move at constant speed. Specifically, at time $t \in [0,1]$, we have the two moving points $\pA(t ) = (1-t)\pA_1 + t \pA_2$ and $\pB(t ) = (1-t)\pB_1 + t \pB_2$.  The cost of this morphing boils down to a function
\begin{equation*}
    f(x) = \sqrt{ a  x^2 + bx + c },
\end{equation*}
for some constant $a,b,c$, and integrating $f$ on an interval $[0,\lenX{\seg}]$. Using the computer algebra system \texttt{maxima} \cite{maxima-23} yields the following indefinite integral (which we subsequently verified is correct).
\begin{equation*}
    F(x)
    =
    \int f(x) \dX{x}
    =%
    \pth{\frac{c}{2 \sqrt{a}} - \frac{b^2}{8a^{3/2}}}
    \asinh \frac{2 a x  +b }{\sqrt{4ac - b^2}}
    + \pth{\frac{x}{2} + \frac{b}{4a} } f(x).
\end{equation*}
In particular, $\priceX{\seg} = F(\lenX{\seg}) - F(0)$ is the price of the \CDTW charged to $\seg$. We repeat the same argument, but now the morphing is interpreted as a point moving on $\seg'$. This would yield a similar indefinite integral $G(\cdot)$, and the price of $\seg'$ is $\priceX{\seg'} = G(\lenX{\seg'}) - G(0)$.

In the free space diagram, the morphing between $\seg$ and $\seg'$ corresponds to a segment that lies in a cell, and $\priceX{\seg} + \priceX{\seg'}$ is the price the sweep distance assigns to this segment in the free space diagram.

\myparagraph{Computing the sweep distance.}  The above gives us a new pricing of the edges of the \VE graph (see \defref{v_e_graph}), as every edge is a segment in the free space diagram. We can then use Dijkstra to compute the shortest path in the \VE graph. Naturally, the resulting path is not monotone, but monotonicity can be easily achieved by introducing middle points, as described in \remref{b_refinement}. We then recompute the shortest path, getting a morphing, and repeat the refinement step if the morphing is not yet monotone. Usually, one round of refinement seems to suffice. We output the computed distance (for the associated monotone morphing).  We refer to this quantity as the \emphi{sweep distance} between the two original curves, denoted by $\dSWY{\cA}{\cB}$.

This readily yields the following.

\begin{lemma}
    The sweep distance between two curves $\cA$ and $\cB$ can be computed by repeatedly running Dijkstra on the appropriately defined \DAG. The sweep distance computed is an upper bound on the \CDTW distance between the two curves.
\end{lemma}

\begin{remark}
    The number of times Dijkstra has to be invoked by the above algorithm for real-world inputs seems to be once or twice, as the Sweep Distance tends to be larger for longer curves (since it's an integral (of a non-negative function) over the curves).
\end{remark}

Given a curve $\cA$, a \emphi{splitting} of $\cA$ is the curve resulting from introducing a vertex in the middle of each segment of $\cA$. Let $\cA^i$ denote the curve resulting from $i$ iterations of splitting. We then have the following.

\begin{lemma}
    $\lim_{i\rightarrow \infty} \dSWY{\cA^i}{\cB^i} = \dTWY{\cA}{\cB}$.
\end{lemma}
\begin{proof}
    As $i$ increases, a cell in the free space diagram of $\cA^i$ and $\cB^i$ becomes smaller, corresponding to the distance functions between shorter subsegments of the two curves. In particular, the functions become closer to being constant on each cell, and the relevant integral, forming the \CDTW distance, is better approximated by the shortest path on the \VE graph.
\end{proof}

\subsubsection{Computing a lower bound on the \CDTW distance using the sweep distance}

It is natural to try to compute a lower bound on the \CDTW distance, so that one can estimate the quality of the upper bound computed by the above algorithm. To this end, consider the free space diagram $\doggy =\doggyY{\curveA}{\curveB}$ induced by $\cA$ and $\cB$, see \defref{f_s_d}.  Let $\Cells$ be the set of grid cells of $\doggy$, and consider the ``silly'' elevation function,
\begin{equation*}
    \forall (x,y) \in \doggy
    \qquad
    \elevC(x,y) = \min_{(x',y') \in \Cells(x,y)} \elevY{x'}{y'},
\end{equation*}
where $\Cells(x,y)$ denotes the cell in the grid of $\doggy$ containing the point $(x,y)$. Namely, we flatten the elevation function inside each grid cell to its minimum. Similarly, we assign a grid edge $e$ of $\doggy$, the minimum of the minimum elevations on the two cells adjacent to $e$.

\myparagraph{Algorithm.}  Observe that the grid of $R$ naturally defines a directed grid graph where each vertical cell edge is directed upwards, and each horizontal cell edge is directed to the right (note that this graph is different than the \VE graph). We now weight the edges of this graph according to the above ``flattened" elevation function.  Now, compute the shortest path on this weighted grid graph going from bottom left to top right. We claim the resulting quantity $\lbSDY{\cA}{\cB}$ is a lower bound on the sweep distance between the two curves.

\begin{lemma}
    $\lbSDY{\cA}{\cB} \leq \dSWY{\cA}{\cB}$.
\end{lemma}

\begin{proof}
    Consider computing the sweep distance for this ``elevation'' function.  Since the sweep distance is decomposed along its $x$ and $y$ components, and the function is a constant inside a grid cell, one can safely assume the optimal morphing is axis-aligned. The only remaining possibility is that the optimal morphing enters a column on a left edge, climbs in a stairway, and leaves through an edge on the right. Let $e_1, e_2, \ldots, e_k$ be the horizontal edges this path intersects, and observe that one can always modify this path to move vertically to the cheapest edge (under the constant elevation function), and cross using this function. Thus, the optimal path in this case can be restricted to use the grid edges, see \figref{snapping}.  The case where the path does the same thing in a row can be handled similarly.
\end{proof}

\begin{figure}
    \centering \includegraphics[page=1]{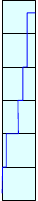}%
    \qquad%
    \includegraphics[page=2]{figs/climb}
    \caption{Snapping the path to the grid.}
    \figlab{snapping}
\end{figure}

\begin{lemma}
    For any two curves $\cA$ and $\cB$, we have $\lim_{i\rightarrow \infty} \lbSDY{\cA^i}{\cB^i} = \dTWY{\cA}{\cB}$.
\end{lemma}
\begin{proof}
    The same integration argument works -- the difference in value between the computed value for $\cA^i$ and $\cB^i$, and the optimal \CDTW distance shrinks as $i$ increases, as the function being integrated ``converges'' to a constant in the grid cells of the refined free space diagram.
\end{proof}

\subsubsection{Discussion of the two algorithms}

Experiments show that the upper bound computed by the algorithm described above is a good approximation to the optimal solution (even with no refinement). This holds because some of the \VE graph edges (specifically, left-to-right and bottom-to-top) realize the minimum lines in the cell of the grid, thus at least one of the two integrals is minimized if the solution follows this segment, see \obsref{min_lines}. Thus, the paths suggested by the \VE graph are already ``cheap'', at least locally.

The lower bound algorithm is far from being very good, as far as convergence after refinement. We leave the question of a better algorithm for computing a lower bound as an open problem for further research.

\subsection{Fast output-sensitive simplification extractor}
\seclab{fast_simplifier}

Given a curve $\cA$, the task at hand is to preprocess it so that, given a parameter $\mu$, one can quickly extract a subcurve $\cB$ of $\cA$, induced by a subset of its vertices, such that the \Frechet distance between $\cB$ and $\cA$ is at most $\mu$. Furthermore, since we are dealing with huge input curves, the query time has to be proportional to $\cardin{\VX{\cB}}$ -- namely, the query time is output sensitive, not depending directly on the input size. We first describe an algorithm that given a curve with $n$ vertices, in $O(n \log n)$ time, computes an array of numbers of size $n$, such that one can extract quickly an output-sensitive simplification from it. We then discuss ways to improve the quality of this simplification.

\subsubsection{Fast extractor}

\myparagraph{A fast approximation.}

For a closed set $X \subseteq \Re^d$ and a point $\pA \in \Re^d$, let $\dSY{\pA}{X} = \min_{ \pC \in X} \dY{\pA}{\pC}$ be the \emphi{distance} of $\pA$ to $X$. Similarly, let $\nnY{\pA}{X} = \arg \min_{ \pC \in X} \dY{\pA}{\pC}$ be the \emphi{nearest point} of $X$ to $\pA$. The \emphi{projection} of a set $Y$ to a closed set $X \subseteq \Re^d$, is the set $\nnY{Y}{X} = \Set{ \nnY{\pA}{X}}{ \pA \in Y}$.

\begin{lemma}
    \lemlab{spine}%
    Given a (directed) curve $\cA = \pA_1 \ldots \pA_n$ with $n$ vertices, consider any (directed) segment $\seg =\pD \pD'$, such that $\nnY{\pA_1}{\seg} = \pD$ and $\nnY{\pA_n}{\seg} = \pD'$. One can compute in linear time a monotone morphing $\mrp$, and its width $\WidthX{\mrp}$, between $\cA$ and $\seg$, such that $\distFr{\cA}{\seg} \leq \WidthX{\mrp} \leq 3 \distFr{\cA}{\seg}$.
\end{lemma}
\begin{proof} %
    Consider first the morphing $\mrp'$, which maps each point of $\cA$ to its nearest point on $\seg$. Clearly, $\WidthX{\mrp'} \leq \WidthX{\mrp^\star}$, where $\mrp^*$ is the optimal (monotone) morphing between $\cA$ and $\seg$.  Consider the monotone version $\mrp$ of $\mrp'$ that one gets from $\mrp'$ by never moving back along $\seg$.  If $\mrp'$ is monotone, then $\mrp = \mrp'$ and $\WidthX{\mrp} = \distFr{\cA}{\seg}$, and we are done.

    \begin{figure}[h]
        \phantom{}\hfill%
        \includegraphics[page=1]{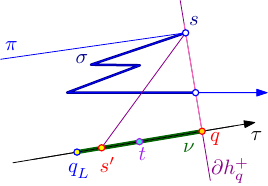}%
        \hfill%
        \includegraphics[page=2]{figs/zig_zag_zog}%
        \hfill\phantom{}%
        \caption{Illustration of the proof of \lemref{spine}.}
        \figlab{z_z_z}
    \end{figure}
    The curve $\cA$ is uniformly parameterized on the interval $I = [0,\lenX{\cA}]$, see \defref{u_parameter}.  The morphing $\mrp'$ matches $\cA(t)$ with the point $\nnY{\cA(t)}{\seg}$.

    Let $f(t) = \lenX{\nnY{\cA(t)}{\seg} - z}$ be the normalized parameterization of this projection point on $\seg$.  Thus, for all $t$, $f(t) \in [0,\lenX{\seg}]$, $f(0) = 0$, and $f( \lenX{\cA}) = \lenX{\seg}$.  The \emph{max} version of $f$ is $F(t) = \max_{x: [0,t]} f(x)$, and the morphing $\mrp$, for all $t$, maps $\cA(t)$ to $\seg(F(t))$.  A maximal interval $J = [\alpha, \beta]\subseteq [0,\lenX{\cA}]$, with $\alpha < \beta$, such that $F(J)$ is a single value, is a \emphw{plateau}. For a plateau, the morphing $\mrp$ matches all the points of $\cA(J)$ to the point $q = \seg( F(J))$. A point $x \in [0,\lenX{\cA}]$ not contained in any plateau, is matched by $\mrp$, to the same point it was matched by $\mrp'$, that is $\seg( f(x) ) = \nnY{\cA(x)}{\seg}$.

    The only places where the leash used by $\mrp$ might be larger than $\mrp^\star$ are in the plateaus.  So consider such a plateau $J = [\alpha,\beta]$, with $\alpha < \beta$. Let $\nu = \Set{ \seg(f(t))}{t \in J}$ be the portion of $\seg$ that is being used by $\mrp'$ for points of $\cB=\cA(J)$. The start of this plateau in $\cA$ is $s=\cA(\alpha)$, and the end is $t = \cA(\beta)$. The left endpoint of $\nu \subseteq \seg$ is the projection of some point $u = \cA(\gamma)$, with $\gamma \in (\alpha, \beta)$, such that $f(J) = [f(\gamma), f(\beta)]$. Let $q_L = \seg(f(\gamma))$ and $q = \seg(f(\beta))$ be the right and left endpoints of $\nu$, respectively, see \figref{z_z_z}. Finally, let $r$ be the middle point of $\nu$.

    If the optimal (monotone) morphing $\mrp^\star$ matches $s$ to a point $s' \in \seg$ that appears before $r$ on $\seg$, then $\WidthX{\mrp^\star} \geq \dY{s}{s'} \geq \lenX{\nu}/2$, see \figref{z_z_z} (left).  Otherwise, $\mrp^\star$ maps $u$ to a point $u' \in \seg$ that appears after $r$ on $\seg$. But then, $\WidthX{\mrp^\star} \geq \dY{u}{u'} \geq \lenX{\nu}/2$, see \figref{z_z_z} (right). We conclude that
    \begin{equation*}
        \WidthX{\mrp^\star}
        \geq%
        \max( \WidthX{\mrp'},  \lenX{\segA}/2 ).
    \end{equation*}
    As a reminder, $\mrp$ matches all the points of $\cB = \cA(J)$ to $q$.  Let
    \begin{math}
        \ell = \max_{\pA \in \cB} \dSY{\pA}{\nu}
    \end{math}
    be the maximum leash length deployed by $\mrp'$ for any point of $\cB$.  We have that the width of $\mrp$ restricted to $\cB$ is
    \begin{align*}
      \ell_J
      &=%
        \max_{\pA \in \cB} \dY{\pA}{q}
        \leq
        \sqrt{\ell^2 + \lenX{\nu}^2}
        \leq
        \sqrt{\WidthX{\mrp'}^2 + \lenX{\nu}^2}
        \leq
        \sqrt{2} \max(\WidthX{\mrp'},   \lenX{\segA} )
      \\&
      \leq
      2\sqrt{2} \max(\WidthX{\mrp'},   \lenX{\segA}/2 )
      \leq
      3\WidthX{\mrp^{\star}}.
    \end{align*}
    Now, $\WidthX{\mrp}$ is the maximum of $\WidthX{\mrp'}$ and the maximum of $\ell_J$, over all plateaus $J$, which implies the claim.
\end{proof}

\myparagraph{Preprocessing.}  The idea is to build a hierarchical representation of the curve. So let the input curve be $\cA = \pA_1 \pA_2 \ldots \pA_n$ -- the output will be an array $A[1\ldots n]$ of real numbers.  The roughest approximation for $\cA$ is the \emphw{spine} $\pA_1 \pA_n$. If we want a finer approximation, the natural vertex to add is $\pA_{\alpha(1,n)}$, where $k = \alpha(i,j) = \floor{(i + j)/2}$, which yields the curve $\pA_1 \pA_{k} \pA_n$.  The recursive algorithm $\Alg(1,n)$ uses the fast approximation algorithm of \lemref{spine} to compute
\begin{equation*}
    A[k] = \max( D(\cA[1, k] ), D(\cA[k,n])].
\end{equation*}
The algorithm now continues filling the array recursively, by calling $\Alg(1,k)$ and $\Alg(k,n)$.

\myparagraph{Extracting the simplification (\extract).}

Given a parameter $w$, the algorithm performs a recursive traversal of the curve, as described above. If the traversal algorithm arrives to an interval $\IRY{i}{j}$, with $k = \alpha(i,j)$, such that $A[k] > w$, then the algorithm adds $\pA_k$ to the simplified curve, after extracting the simplification recursively on the range $\IRY{i}{k}$, and before extracting the simplification recursively on $\IRY{k}{j}$. If $A[k] \leq w$, then the algorithm adds $\pA_k$ to the simplification, without performing the recursive calls.

\begin{lemma}
    \lemlab{fast_simplifier}%
    Given a polygonal curve $\cA$ with $n$ vertices, the algorithm \Alg preprocess it, in $O(n \log n)$ time, such that given a parameter $w \geq 0$, the query algorithm \extract computes a subcurve $\cAs$ (induced by a subset of the vertices of $\cA$), such that $\distFr{\cA}{\cAs} \leq w$. The extraction algorithm works in $O( \cardin{ \VX{\cAs}} )$ time.
\end{lemma}
\begin{proof}
    The running time bounds are immediate. As for the correctness, it follows by observing that the simplification being output breaks the input curve into sections, such that each section is matched to a segment, such that the \Frechet distance between a section and its corresponding segment is at most $w$.
\end{proof}

\begin{remark}
    Note that the above extracts a simplification quickly -- this simplification might potentially have many more vertices than necessary, but in practice, it works pretty well. We can apply the algorithm described next as a post-processing stage to reduce the number of vertices.
\end{remark}

\subsubsection{A greedy simplification.}

A better simplification algorithm, that yields fewer vertices than the $\delta$-simplification used above, and the fast extractor used above, is to do a greedy approximation as suggested by Aronov \etal \cite{ahkww-fdcr-06}. Let $\cA = \pA_1 \ldots \pA_n$. The algorithm initially starts at the vertex $\pA_1$. Assume that it is currently at the vertex $\pA_j$.
The algorithm then sets the next vertex in the simplification to be $p_k$, where $k$ is the smallest index greater than $j$ such that $\distFr{\cA[\pA_j, \pA_k]}{\pA_j\pA_k} \leq \delta$ and $\distFr{\cA[\pA_j, \pA_{k+1}]}{\pA_j\pA_{k+1}}> \delta$ (or $k=n$ if the end of the curve is reached).  One can use the algorithm \lemref{spine} coupled with exponential and binary search to compute $k$. This takes $O( (k-j+1) \log (k-j+1))$ time (of course, here the guarantee is somewhat weaker). This algorithm yields quite a good approximation in practice with fewer vertices than the $\delta$-simplification.

\begin{lemma}
    \lemlab{greedy}%
    Given a curve $\cA$ with $n$ vertices, and a parameter $\delta$, the above algorithm computes, in $O(n \log n)$ time, a curve $\cAs$, such that $\distFr{\cA}{\cAs} \leq \delta$.
\end{lemma}

To get a better, simplified curve quickly, after preprocessing, given a parameter $\delta \geq 0$, one can first use the fast extractor (\extract) to compute a simplification (with distance $\leq \delta/10$), and then apply the above algorithm to the resulting curve to compute a simplification with distance $\leq 0.9\delta$. By the triangle inequality, the resulting simplification is distance at most $\delta$ from the original curve.  This combined simplification approach works quite well in practice.

\begin{remark}
    Somewhat confusingly, we have three simplification algorithms described in this paper. The first, \algref{simplification}, is a simple linear scan of the input curve, and is potentially slow if the input curve is large, particularly if we need to generate many simplifications of the same curve. Furthermore, in practice, the upper bound it provides on the error is way bigger than the actual error. The second approach, using \Alg for preprocessing and using \extract, is much faster but still yields (in practice) inferior simplifications. Using the greedy simplification algorithm enables us to ``cleanup'', the simplification and get a curve with significantly fewer vertices.
\end{remark}

\section{Conclusions and future work}

In this work, we have demonstrated that our variant of the \Frechet{} distance is both theoretically efficient and viable in practice. We provided standard libraries for the algorithms in \Julia and \Python. Furthermore, our libraries compute the retractable \Frechet distance/morphing, which seems to perform quite well in practice, and is better at handling noise than the original \Frechet distance.  This suggests that our algorithm may be helpful in applications where \Frechet{} distance computations dominate the runtime, especially when approximation is acceptable.  More generally, we believe the high quality of our implementations, and their easy availability via standard libraries in \Julia and \Python, should make our code and algorithms easily accessible for casual users.

\myparagraph{Clustering.}

Using the implementation given by Bringmann \etal \cite{bkn-wdfpa-21}, there has been substantial empirical work on efficient clustering of curves. Notably, Buchin \etal \cite{bdvn-kcbct-19} empirically study algorithms for $(k,\ell)$-center clustering. We omit the formal definition of this problem here, but of note is the large number of \Frechet{} distance computations required to determine the furthest curve from the set of existing curves chosen as centers. We believe that our algorithm is very well suited for this task, as our heap-based implementation can be easily modified to stop early if the \Frechet{} distance goes beyond a given threshold.

Similarly, there has been prior empirical work on $(k,\ell)$-median clustering \cite{bbknp-klmct-20} under the dynamic time warping distance, and we believe our algorithm would perform well in this application as well.

\myparagraph{Polygonal hierarchies and fast simplification.}

A strategy that works quite well for working on a large database of curves is to precompute, for all input curves, the fast simplification extractor of \lemref{fast_simplifier}. Then, given a query, one can simplify the two curves quickly to the correct resolution $\delta$ (initially guessed to be some large value) using this extractor, then apply the algorithm of \lemref{greedy} to get a simplification with few vertices. With these simplifications, then compute the \Frechet distance, and repeat the process, with smaller $\delta$, till the query is resolved. By caching the simplifications of the input curves, one can achieve better speedup if there are multiple computations involving the same curve.

\myparagraph{When our algorithm performs badly?}

Our algorithm performs poorly when many refinement rounds are required, and even worse when there are many different paths that are more or less equivalent, and the refinement causes the morphing to jump around the free space during refinement. Specifically, such bad examples are made out of ``one-dimensional'' zigzags where the distance is large, just because one of the curves has one more zigzag than the other, or even if the zigzags match but are unfavorable. Such an example also implies that many rounds of refinements are required before the algorithm gets even remotely close to the optimal \Frechet distance. See \figref{bad}. Of course, such inputs are not in any way real-world input, and represent the gap between the theory and practice of computing the \Frechet distance.

\begin{figure}[h]%
    \centerline{%
       \includegraphics[width=0.6\linewidth]%
       {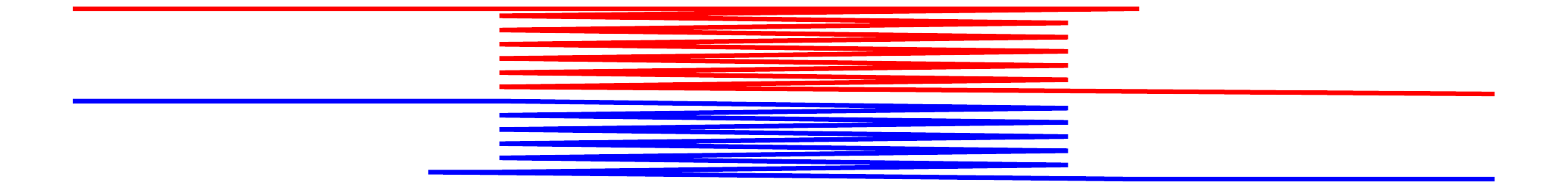}%
    }%
    \vspace*{-0.25cm}%
    \caption{When the new algorithm does badly -- the implementation required about 120 refinement rounds before converging, as the morphing kept leaping around the free space diagram. See \href{\baseUrlX{35}}{here} for more details and animations.  This example can be made worse by scaling the $y$-axis further down.  }%
    \figlab{bad}
\end{figure}

\printbibliography

\appendix

\section{The elevation function}

We need some standard properties of the elevation function. We prove some of them here, so that our presentation would be self-contained.

\subsection{A helper lemma}

\begin{lemma}
    \lemlab{norm_of_affine}%
    for $f: \Re^k \rightarrow \Re^d$ an affine function, the function $u(\pB) = \norm{f(\pB)}$ is convex.
\end{lemma}
\begin{proof}
    The proof is straightforward, and the reader is encouraged to skip reading it.  Fix any two points ${\pnt'}. {\pntA'} \in \Re^k$, and consider the segment ${\pnt'} {\pntA'}$. We need to prove that $h(t) = \norm{ \bigl. f\pth{(1-t){\pnt'} + t {\pntA'}}}$ is convex.  Let $\pnt = f\pth{{\pnt'}}$ and $\pntA = f\pth{{\pntA'}}$.  Since $f$ is affine, we have that
    \begin{equation*}
        h(t)%
        =%
        \norm{ \bigl. f\pth{(1-t) {\pnt'} + t {\pntA'}}}%
        =%
        \norm{ \bigl. (1-t)f\pth{{\pnt'}} + t f\pth{{\pntA'}}}%
        =%
        \norm{ (1-t) \pnt + t \pntA }%
        =%
        \sqrt{ \sum\nolimits_{i=1}^d g_i(t)},
    \end{equation*}
    where
    \begin{math}
        g_i(t)%
        =%
        \pth{ (1-t) p_i + t q_i }^2%
        =%
        \alpha_i t^2 + \beta_it + \gamma_i,
    \end{math}
    $\alpha_i, \beta_i, \gamma_i$ are constants, for $i=1,\ldots, d$, $\pnt = (p_1,\ldots, p_d)$, and $\pntA = (q_1,\ldots, q_d)$. For any $i$, the function $g_i(t)$ is nonnegative. If $\alpha_i = 0$ then $g_i(t) = \gamma_i$ and then $p_i = q_i$.  Otherwise, $\alpha_i >0$ and $g_i(t)$ is a parabola.  Let $\alpha = \sum_i \alpha_i$, $\beta=\sum_i \beta_i$ and $\gamma = \sum_{i} \gamma_i$.  Consider the function $g(t) = \sum_{i=1}^d g_i(t) = \alpha t^2 + \beta t + \gamma$. If, for all $i$, $p_i = q_i$, then $g(t) = \gamma$, and the claim trivially holds. Otherwise, $g(t)$ is a non-negative parabola with $\alpha > 0$, and since it has at most a single root, we have that $\beta^2 - 4 \alpha \beta \leq 0$.

    Now, we have
    \begin{align*}
      &h'(t) = \frac{2\alpha t + \beta}{2 \sqrt{ \alpha t^2 + \beta t +
        \gamma}} = \frac{h_1(t)}{h(t)}%
        \qquad
        \text{for } h_1(t) = \alpha t + \beta/2, \\
      \text{ and } \qquad
      &%
        h''(t) %
        =%
        \frac{h(t)g'(t) - h'(t) h_1(t)}{(h(t))^2}%
        =%
        \frac{\alpha h(t) - (h_1(t))^2/h(t)}{(h(t))^2} %
        =%
        \frac{(h(t))^2 - (h_1(t))^2/\alpha}{(h(t))^3/\alpha}.
    \end{align*}
    This implies that
    \begin{align*}
      \signX{h''(t)}%
      &%
        = \signX{(h(t))^2 - (h_1(t))^2/\alpha} = \signX{
        \alpha t^2 + \beta t + \gamma - \alpha t^2 - \beta t -
        \beta^2 /4\alpha} %
      \\& = %
      \signX{\gamma -\beta^2 /4\alpha } = %
      \signX{4\alpha\gamma -\beta^2 } \geq 0,
    \end{align*}
    since $\alpha > 0$ and $\beta^2 -4\alpha \gamma \leq 0$. Which implies that $h(\cdot)$ is convex, and so is $\norm{f(\cdot)}$.
\end{proof}

\subsection{Back to the elevation function}

We are given two segments $\pps \ppe$ and $\qqs \qqe$, where $L_\pp = \dY{\pps}{\ppe}$ and $L_\qq = \dY{\qqs}{\qqe}$.  Their uniform parameterization is
\begin{equation*}
    f(s)
    =%
    \pps + s \pp
    \qquad\text{and}\qquad%
    g(t)
    =%
    \qqs + t \qq,
\end{equation*}
where $\pp = (\ppe - \pps)/L_\pp$ and $\pp = (\ppe - \pps)/L_\qq$.  The elevation function is
\begin{equation*}
    \forall (s,t) \in [0,L_p] \times [0,L_q] \qquad
    \elevY{s}{t}%
    =%
    \norm{ \uus +  s\pp - t\qq},
\end{equation*}
where $\uus = \pps - \qqs$.  Since the elevation function is the norm of an affine function, it is convex, by \lemref{norm_of_affine}.

\begin{figure}[b]
    \phantom{}%
    \hfill%
    \includegraphics[width=0.25\linewidth]{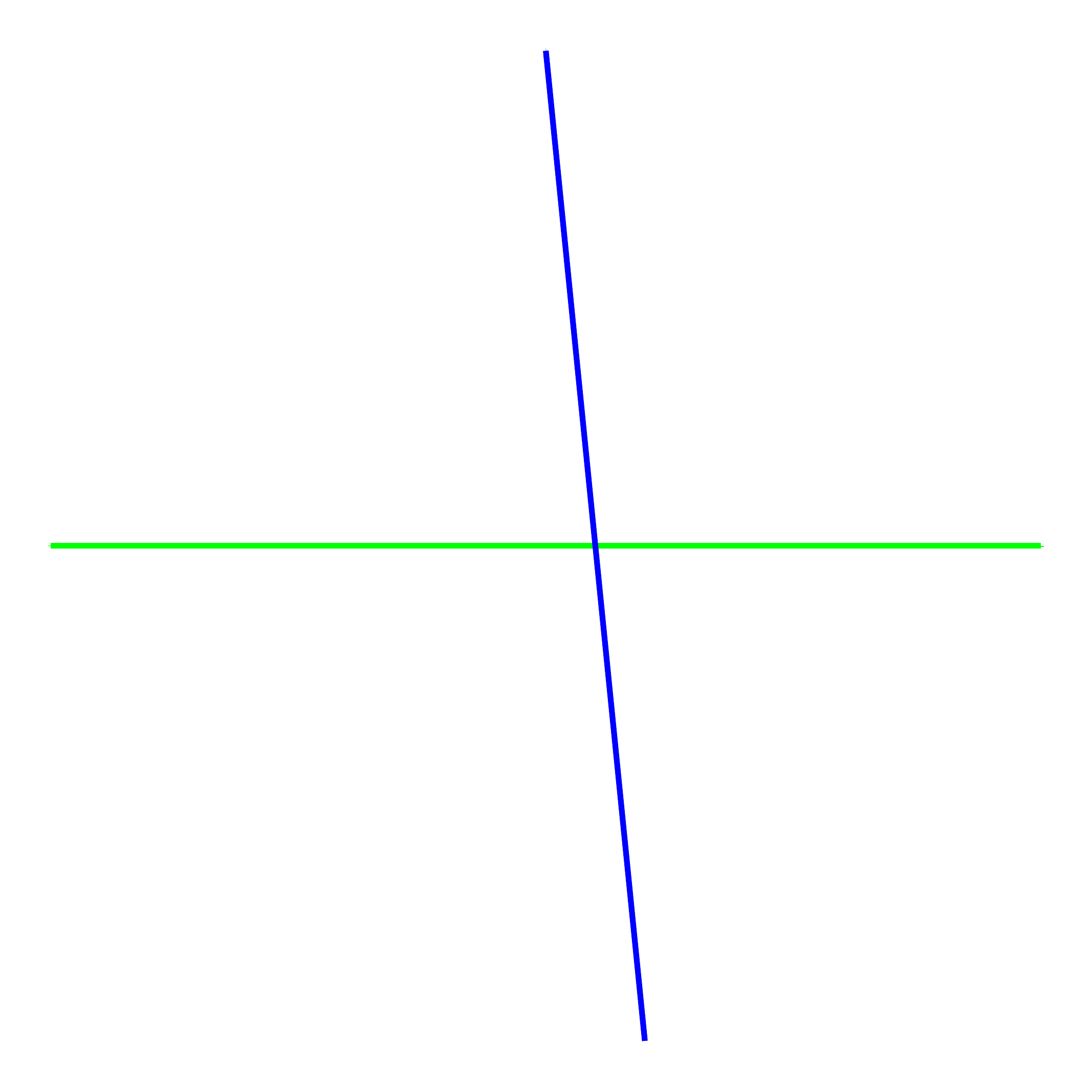} \hfill%
    \includegraphics{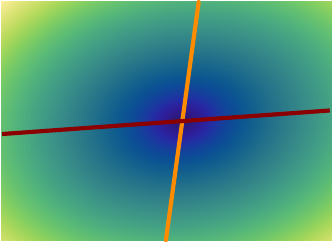} \hfill%
    \phantom{}%
    \caption{The elevation function for two segments, and the minimum lines.  See relevant animations \href{\baseUrlX{18}}{here}. %
    }%

    \figlab{18}
\end{figure}

The squared elevation function is
\begin{align*}
  E(s,t)
  =%
  \pth{\elevY{s}{t}}^2%
  &=%
    \DotProd{\uus + s \pp - t \qq}{\uus + s \pp - t \qq}
  \\&
  =
  \norm{\uus}^2 + 2 \DotProd{\uus}{s \pp - t \qq} +
  s^2  - 2 st \DotProd{\pp}{\qq} + t^2,
\end{align*}
since $\norm{\pp}^2 = 1$ and $\norm{\qq}^2 = 1$.  Its \emphi{derivative} along $s$ and $t$ respectively is
\begin{align*}
  \partial_s {E}(s,t)
  &=
    2  \DotProd{\uus}{\pp} +
    2s  - 2 t \DotProd{\pp}{\qq}
  \\%
  \text{and}\qquad%
  \partial_t {E}(s,t)
  &=
    -2 \DotProd{\uus}{  \qq}
    - 2 s \DotProd{\pp}{\qq} + 2t
\end{align*} %
In particular, let $h(\alpha)$ (resp. $v(\beta)$) be the $s$ (resp. $t$) coordinate of the minimum of the elevation function on the horizontal line $t = \alpha$ (resp. vertical line $s=\beta$). We have that $h(\alpha)$ is the solution to the (linear) equation $ \partial_s {E}(s,\alpha) = 0$ (resp.  $ \partial_t {E}(\beta,t) = 0$). That is
\begin{equation*}
    h(\alpha)
    =%
    \DotProd{\pp}{\qq} \alpha
    -
    {\DotProd{\uus}{\pp}}
    \qquad\text{and}\qquad%
    v(\beta)
    =%
    \beta \DotProd{\pp}{\qq}
    + \DotProd{\uus}{  \qq}.
\end{equation*}

\begin{observation}
    \obslab{min_lines}%
    In particular, the edges connecting the left portal to the right portal, and the edge connecting the bottom portal of a cell to its top portal, are both tracing these minimum edges.
\end{observation}

\end{document}